\documentclass[12pt]{article} \usepackage{fullpage}
\usepackage{amssymb}
\usepackage{ifpdf}
\usepackage[latin1]{inputenc} 
\usepackage{graphicx} \usepackage{amsmath} 
\usepackage{algorithmic}

\usepackage{ifpdf}

\usepackage{hyperref} 
\usepackage{authblk}

\newtheorem{theorem}{Theorem} 
\newtheorem{cor}[theorem]{Corollary}

\newtheorem{lem}[theorem]{Lemma}
\newtheorem{lemma}[theorem]{Lemma}
\newtheorem{defn}[theorem]{Definition}

\newtheorem{fact}{Fact}

\newenvironment{remark}[1][Remark]{\begin{trivlist} \item[\hskip \labelsep
{\bfseries #1}]}{\end{trivlist}}

\newenvironment{proof}{\noindent {\bf Proof:} \hspace{.4em}}
                      {\hspace{\fill}{$\blacksquare$} \smallskip}

\newenvironment{packed_item}{ \begin{itemize}
  \setlength{\itemsep}{1pt} \setlength{\parskip}{0pt}
  \setlength{\parsep}{0pt}
}{\end{itemize}}

\newcommand{\ignore}[1]{{}}

 \newcommand{\ra}{\rightarrow}

\newcommand{\bt}{\ensuremath{BT_d^h(k)}}
\newcommand{\ft}{\ensuremath{FT_d^h(k)}}
\newcommand{\bth}{\ensuremath{BT_d(h,k)}}
\newcommand{\fth}{\ensuremath{FT_d(h,k)}}

\newcommand{\bttwothree}{\ensuremath{BT_2^3(k)}}

\newcommand{\bttwofour}{\ensuremath{BT_2^4(k)}}
\newcommand{\bttwoh}{\ensuremath{BT_2^h(k)}}

\newcommand{\btdthree}{\ensuremath{BT_d^3(k)}}

\newcommand{\bthk}{\ensuremath{BT_d(h,k)}}
\newcommand{\fthk}{\ensuremath{FT_d(h,k)}}
\newcommand{\dFstate}{\ensuremath{\mathsf{\#detFstates}^h_d(k)}}
\newcommand{\dFdthreestate}{\ensuremath{\mathsf{\#detFstates}^3_d(k)}}

\newcommand{\ndFstate}{\ensuremath{\mathsf{\#ndetFstates}^h_d(k)}}
\newcommand{\dBstate}{\ensuremath{\mathsf{\#detBstates}^h_d(k)}}
\newcommand{\dBdthreestate}{\ensuremath{\mathsf{\#detBstates}^3_d(k)}}

\newcommand{\ndBstate}{\ensuremath{\mathsf{\#ndetBstates}^h_d(k)}}

\newcommand{\nddthreeBstate}{\ensuremath{\mathsf{\#ndetBstates}^3_d(k)}}
\newcommand{\ndtwothreeBstate}{\ensuremath{\mathsf{\#ndetBstates}^3_2(k)}}

\newcommand{\Bpebbles}{\ensuremath{\mathsf{\#Bpebbles}}}
\renewcommand{\Bpebbles}{\ensuremath{\mathsf{\#pebbles}}}

\newcommand{\BWpebbles}{\ensuremath{\mathsf{\#BWpebbles}}}
\newcommand{\FRpebbles}{\ensuremath{\mathsf{\#FRpebbles}}}

\newcommand{\acone}{\ensuremath{\mathbf{AC}^1}}
\newcommand{\aczsix}{\ensuremath{\mathbf{AC}^0(6)}}

\newcommand{\ncone}{\ensuremath{\mathbf{NC}^1}}
\newcommand{\nctwo}{\ensuremath{\mathbf{NC}^2}}
\newcommand{\p}{\ensuremath{\mathbf{P}}}
\newcommand{\nl}{\ensuremath{\mathbf{NL}}}
\newcommand{\lspace}{\ensuremath{\mathbf{L}}}
\newcommand{\logcfl}{\ensuremath{\mathbf{LogCFL}}}
\newcommand{\logdcfl}{\ensuremath{\mathbf{LogDCFL}}}
\newcommand{\np}{\ensuremath{\mathbf{NP}}}
\newcommand{\ph}{\ensuremath{\mathbf{PH}}}

\newcommand{\func}{\ensuremath{Children_d^h(k)}}
\newcommand{\functwofour}{\ensuremath{Children_2^4(k)}}
\newcommand{\funcmod}{\ensuremath{\mathit{SumMod}_d^h(k)}}
\newcommand{\funcmodtwothree}{\ensuremath{\mathit{SumMod}_2^3(k)}}

\newcommand{\nkway}{\ensuremath{N_{\mbox{\scriptsize nondet}}^{\mbox{\scriptsize $k$-way}}}}
\newcommand{\dkway}{\ensuremath{N_{\mbox{\scriptsize det}}^{\mbox{\scriptsize $k$-way}}}}
\newcommand{\ntwoway}{\ensuremath{N_{\mbox{\scriptsize
        nondet}}^{\mbox{\scriptsize $2$-way}}}}

\newcommand{\Vars}{\mathsf{Vars}}   
\newcommand{\var}{{\sf var}}

\newcommand{\adv}{{\sf adv}}

\newcommand{\Dom}{{\sf Dom}}

\newcommand{\isc}{i_{{\sf sc}}}

\def\IntAdv{\text{{\sc InterAdv}}}
\newcommand{\DEF}{{\downarrow}}
\newcommand{\UL}{U_{\mathsf{L}}}
\newcommand{\node}{{\sf node}}

\def\titlefootnote{\footnote{Versions of parts of this paper appeared in
\cite{BCMSW09} and \cite{BCMSWA09}}}

\begin{document} 

\title{Pebbles and Branching Programs for Tree Evaluation\titlefootnote}
\author[1]{Stephen Cook}
\author[2]{Pierre McKenzie}
\author[1]{Dustin Wehr}
\author[3]{Mark Braverman}
\author[4]{Rahul Santhanam}
\affil[1]{University of Toronto}
\affil[2]{Université de Montréal}
\affil[3]{Microsoft Research}
\affil[4]{University of Edinburgh}


\maketitle
\begin{abstract}
We introduce the {\em tree evaluation problem},
show that it is in \logdcfl\
(and hence in {\bf P}), and study its branching program complexity
in the hope of eventually proving a superlogarithmic space lower bound.
The input to the problem is a rooted, balanced $d$-ary tree of height $h$,
whose internal nodes are labeled with $d$-ary functions
on $[k]=\{1,\ldots,k\}$, and whose leaves are labeled with elements
of $[k]$.  Each node obtains a value in $[k]$ equal to its $d$-ary
function applied to the values of its $d$ children.  The output is the
value of the root.  We show that the standard black pebbling
algorithm applied
to the binary tree of height $h$ yields a deterministic $k$-way
branching program with $O(k^h)$ states solving this problem,
and we prove that this upper bound is tight
for $h=2$ and $h=3$.  We introduce
a simple semantic restriction called {\em thrifty}
on $k$-way branching programs
solving tree evaluation problems and show that the same state
bound of $\Theta(k^h)$ is tight
for all $h\ge 2$ for deterministic
thrifty programs.  We introduce fractional pebbling for trees
and show that this yields nondeterministic thrifty programs
with $\Theta(k^{h/2+1})$ states solving the Boolean problem
``determine whether the root has value 1'', and prove
that this bound is tight for $h=2,3,4$.  We also prove that
this same bound is tight for unrestricted
nondeterministic $k$-way branching programs solving the Boolean
problem for $h=2,3$.
\end{abstract}

\newpage

\tableofcontents

\newpage

\section{Introduction}

Below is a nondecreasing sequence of standard complexity classes
between \aczsix\ and the polynomial hierarchy.
\begin{equation}\label{classes}
  \aczsix \subseteq \ncone \subseteq \lspace \subseteq \nl \subseteq 
    \logcfl 
   \subseteq \acone \subseteq \nctwo \subseteq \p \subseteq \np
   \subseteq \ph
\end{equation}
A problem in \aczsix\ is given by a uniform family of polynomial
size bounded depth circuits with unbounded fan-in Boolean and
mod 6 gates.  As far as we know an \aczsix\ circuit cannot determine
whether a majority of its input bits are ones, and yet we
cannot provably separate \aczsix\ from any of the other classes in the
sequence.  This embarrassing state of affairs motivates this paper
(as well as much of the lower bound work in complexity theory).

We propose a candidate for separating \nl\ from \logcfl.
The \emph{Tree Evaluation problem} \fth\ is defined as follows.
The input to \fth\ is a balanced $d$-ary tree of height $h$,
denoted $T^h_d$ (see Fig. \ref{sample}).  Attached
to each internal node $i$ of the tree is some explicit function
$f_i: [k]^d\rightarrow [k]$ specified as $k^d$ integers in
$[k]=\{1,\ldots,k\}$.  Attached to each leaf is a number in $[k]$.
Each internal tree node takes a value in $[k]$ obtained by applying its
attached function to the values of its children.  The function
problem \fth\ is to compute the value of the root, and the Boolean
problem \bth\ is to determine whether this value is $1$.

\ifpdf
\else
\begin{figure}
\vspace*{.3cm}
\hspace*{3.8cm} \includegraphics[scale=0.40]{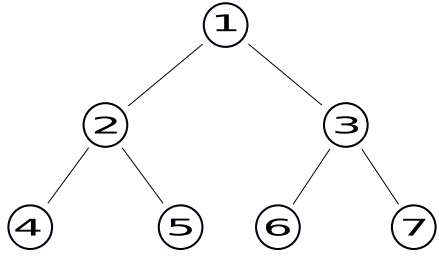}
\hspace*{1cm} \vspace*{.2cm}
\caption{A height 3 binary tree $T_2^3$ with nodes numbered heap style.}
\label{sample}
\end{figure}
\fi

It is not hard to show that
a deterministic logspace-bounded polytime auxiliary
pushdown automaton decides \bth, where $d$,$h$ and $k$ are input
parameters. This implies by \cite{su78}
that \bth\ belongs to the class \logdcfl\ of languages logspace
reducible to a deterministic context-free language.  The latter class
lies between \lspace\ and \logcfl,
but its relationship with \nl\ is unknown (see \cite{ma07} for a recent
survey).  We conjecture that \bth\ does not lie in \nl.
A proof would separate \nl\ and \logcfl, and hence (by (\ref{classes}))
separate \ncone\ and \nctwo.

Thus we are interested in proving superlogarithmic
space upper and lower bounds (for fixed degree $d\ge 2$)
for \bth\ and \fth.  Notice that for each constant $k=k_0\ge 2$,
$BT_d(h,k_0)$ is an easy generalization of the Boolean formula value
problem for balanced formulas, and hence it is in \ncone\
and \lspace.  Thus it is important that $k$ be an unbounded
input parameter.

We use branching programs (BPs) as a nonuniform
model of Turing machine space:  A lower bound of $s(n)$ on the
number of BP states implies a lower bound of $\Theta(\log s(n))$
on Turing machine space, but to prove the converse we would need
to supply the machine with an advice string for each input length.
Thus BP state lower bounds are stronger than TM space
lower bounds, but we do not know how to take advantage of the
uniformity of TMs to get the supposedly easier lower bounds on TM space.
In this paper all of our lower bounds are nonuniform and all of
our upper bounds are uniform.

In the context of branching programs we think of $d$ and $h$ as
fixed, and we are interested in how the number of states required
grows with $k$.   To indicate this point of view we write the
function problem \fth\ as \ft\ and the Boolean problem \bth\ as \bt.
For this it turns out that $k$-way BPs are
a convenient model, since an input for \bt\ or \ft\ is naturally
presented as a tuple of elements in $[k]$.  Each nonfinal state
in a $k$-way BP queries a specific element of the tuple, and branches
$k$ possible ways according to the $k$ possible answers.

It is natural to assume that the inputs to Turing machines
are binary strings, so 2-way BPs are a closer model of TM space than
are $k$-way BPs for $k>2$. 
But every 2-way BP is easily converted to a $k$-way
BP with the same number of states, and every $k$-way BP can be
converted to a 2-way BP with an increase of only a factor of $k$
in the number of states,
so for the purpose of separating \lspace\ and \p\ we may as
well use $k$-way BPs.

Of course the number of states required by a $k$-way BP
to solve the Boolean
problem \bt\ is at most the number required to solve the function
problem \ft.  In the other direction it is easy to see
(Lemma \ref{l:FvsB}) that \ft\ requires at most a factor of $k$
more states than \bt.  From the point of view of separating \lspace\
and \p\ a factor of $k$ is not important.  Nevertheless it is
interesting to compare the two numbers, and
in some cases (Corollary \ref{c:HtThree}) we can prove tight
bounds for both:  For deterministic BPs solving height 3 trees
they differ by a factor of $\log k$ rather than $k$.

The best (i.e. fewest states) algorithms that we know
for deterministic $k$-way
BPs solving \ft\ come from black pebbling algorithms for trees:
If $p$ pebbles suffice to pebble
the tree $T^h_d$ then $O(k^p)$ states suffice for a BP to solve
\ft\ (Theorem \ref{t:pebSim}). 
This upper bound on states is tight (up to
a constant factor) for trees of height $h=2$ or $h=3$
(Corollary \ref{c:HtThree}), and we
suspect that it may be tight for trees of any height.

There is a well-known generalization of black pebbling called
black-white pebbling which naturally simulates nondeterministic
algorithms.  Indeed if $p$ pebbles suffice to black-white pebble
$T^h_d$ then $O(k^p)$ states suffice for a nondeterministic BP
to solve \bt.  However the best lower bound we can obtain
for nondeterministic BPs solving \bttwothree\
(see Figure \ref{sample}) is $\Omega(n^{2.5})$, whereas
it takes 3 pebbles to black-white pebble the tree $T^3_2$.
This led us to rethink the upper bound, and we discovered that
there is indeed a nondeterministic BP with $O(k^{2.5})$ states
which solves \bttwothree.  The algorithm comes from a black-white
pebbling of $T^3_2$ using only 2.5 pebbles: It places a half-black
pebble on node 2, a black pebble on node 3, and adds a half white pebble
on node 2, allowing the root to be black-pebbled
(see Figure \ref{f:bin_h3_fract_ub} on page \pageref{f:bin_h3_fract_ub}).

This led us to the idea of fractional pebbling in general,
a natural generalization of black-white pebbling.  A fractional
pebble configuration on a tree assigns two nonnegative real numbers 
$b(i)$ and $w(i)$ totalling at most 1, to each node $i$ in the
tree, with appropriate rules for removing and adding pebbles.
The idea is to minimize the maximum total pebble weight on the
tree during a pebbling procedure which starts and ends with no
pebbles and has a black pebble on the root at some point.

It turns out that nondeterministic
BPs nicely implement fractional pebbling procedures:  If $p$
pebbles suffice to fractionally pebble $T^h_d$ then $O(k^p)$
states suffice for a nondeterministic BP to solve \bt.
After much work we have not been able to improve upon
this $O(k^p)$ upper bound for any $d,h\ge 2$.  We prove it is optimal
for trees of height 3 (Corollary \ref{c:HtThree}).

We can prove that for fixed degree $d$ the number of pebbles
required to pebble (in any sense) the tree $T^h_d$ grows
as $\Theta(h)$, so the $p$ in the above best-known upper bounds of
$O(k^p)$ states grows as $\Theta(h)$. This and the following fact
motivate further study of the complexity of \ft.

\medspace

\noindent
\begin{fact}\label{f:unbounded}
A lower
bound of $\Omega(k^{r(h)})$ for {\em any} unbounded function $r(h)$
on the number of states required to solve \ft\
implies that $\lspace \ne \logcfl$ (Theorem \ref{t:logDCFL} and
Corollary \ref{c:thegoal}).
\end{fact}

\medspace

Proving tight bounds on the number of pebbles required to fractionally
pebble a tree turns out to be much more difficult than
for the case of whole black-white pebbling.   However we can prove
good upper and lower bounds.
For binary trees of any height $h$ we prove an upper
bound of $h/2 + 1$ and a lower bound of $h/2-1$ (the upper bound
is optimal for $h\le 4$).  These bounds can be generalized to
$d$-ary trees (Theorem \ref{t:daryFract}).

We introduce a natural semantic restriction on BPs which solve \bt\ or
\ft:  A $k$-way BP is {\em thrifty} if it only queries the
function $f(x_1,\ldots,x_d)$ associated with a node when
$(x_1,\ldots,x_d)$ are the correct values of the children of the node.

It is not hard to see that the deterministic BP algorithms that
implement black pebbling are thrifty.  With some effort we were
able to prove a converse (for binary trees):  If $p$ is the
minimum number of pebbles required to black-pebble $T^h_2$ then every
deterministic thrifty BP solving $BT^h_2(k)$ (or $FT^h_2(k)$) requires
$\Omega(k^p)$ states.  Thus any deterministic BP solving these
problems with fewer states
must query internal nodes $f_i(x,y)$ where $(x,y)$ are not the
values of the children of node $i$.  For the decision problem
$BT^h_2(k)$ there is indeed a nonthrifty deterministic BP
improving on the bound by a factor of $\log k$ (Theorem \ref{t:BPUpper}
(\ref{e:dBUpper})), and this is tight for $h=3$
(Corollary \ref{c:HtThree}).  But we have not been able to improve
on thrifty BPs for solving any function problem \ft.

The nondeterministic BPs that implement fractional pebbling are
indeed thrifty.  However here the converse is far from clear:
there is nothing in the definition of {\em thrifty} that hints
at fractional pebbling.   We have been able to prove that thrifty
BPs cannot beat fractional pebbling for binary trees of height
$h=4$ or less, but for general trees this is open.

It is not hard to see that for black pebbling, fractional pebbles
do not help.  This may explain why we have been able to prove
tight bounds for deterministic thrifty BPs for all binary trees,
but only for trees of height 4 or less for nondeterministic
thrifty BPs.

We pose the following as another interesting open question:

\medskip

\noindent
\label{thriftyH}
{\bf Thrifty Hypothesis:}  Thrifty BPs are optimal among $k$-way
BPs solving \ft.

\medskip

Proving this for deterministic BPs would show $\lspace \ne \logdcfl$,
and for nondeterministic BPs would show $\nl\ne\logcfl$.
Disproving this would provide interesting new space-efficient algorithms
and might point the way to new approaches for proving lower bounds.

The lower bounds mentioned above for unrestricted branching programs
when the tree heights are small are obtained in two ways: First using
the Ne\u{c}iporuk method \cite{ne66}, and second using
a method that analyzes the state sequences of the BP computations.
Using the state
sequence method we have not yet beat the $\Omega(n^2)$ deterministic
branching program size barrier (neglecting log factors) inherent to
the Ne\u{c}iporuk method for Boolean problems,
but we can prove lower bounds for 
function problems which cannot be matched by the Ne\u{c}iporuk method
(Theorems \ref{t:rootfunction}, \ref{t:lasttheorem}, \ref{t:childLB},
\ref{t:beatittwice}).
For nondeterministic branching programs with
states of unbounded outdegree, we show that both methods yield a lower
bound of $\Omega(n^{3/2})$ states (neglecting logs) for the decision
problem $BT_2^3$, and this improves on the former
$\Omega(n^{3/2})$ bound obtained for the number of edges
\cite{pu87,ra91} in such BPs.


\subsection{Summary of Contributions}
\begin{itemize}
\item
We introduce a family of computation problems \ft\ and \bt,
$d,h \ge 2$, which we propose as good candidates for separating
\lspace\ and \nl\ from apparently larger complexity classes
in (\ref{classes}).  Our goal is to prove space lower bounds
for these problems by proving state lower bounds for $k$-way
branching programs which solve them.  For $h=3$ we can prove
tight bounds for each $d\ge 2$ on the number of states required by
$k$-way BPs to solve them, namely (from Corollary \ref{c:HtThree})
\begin{align*}
   & \Theta(k^{(3/2)d - 1/2})
\mbox{ for nondeterministic BPs solving $BT^3_d(k)$}\\
   & \Theta(k^{2d-1}/\log k)
\mbox{ for deterministic BPs solving $BT^3_d(k)$}\\
   & \Theta(k^{2d-1})
\mbox{ for deterministic BPs solving $FT^3_d(k)$}
\end{align*}

\item
We introduce a simple and natural restriction called {\em thrifty}
on BPs solving \ft\ and \bt.  The best known upper bounds for
deterministic BPs solving \ft\ and for nondeterministic BPs
solving \bt\ are realized by thrifty BPs.  Proving even much
weaker lower bounds than these upper bounds for unrestricted
BPs would separate \lspace\ from \logcfl\ (see Fact \ref{f:unbounded}
above).  We prove that for binary trees deterministic thrifty BPs cannot
do better than implement black pebbling (this is far from obvious).

\item
We formulate the {\bf Thrifty Hypothesis} (see above).  Either
a proof or a disproof would have interesting consequences.

\item
We introduce {\em fractional pebbling} as a natural generalization
of black-white pebbling for simulating nondeterministic space
bounded computations.  We prove almost tight lower bounds for
fractionally pebbling binary trees (Theorem \ref{t:daryFract}).
The best known upper bounds for nondeterministic BPs solving
\ft\ come from fractional pebbling, and these
can be implemented by thrifty BPs.
An interesting open question is to prove that nondeterministic thrifty
BPs cannot do better than implement fractional pebbling.  (We prove this
for $h=2,3,4$.)

\item
We use a ``state sequence'' method for proving size
lower bounds for branching programs solving \ft\ and \bt, and
show that it improves on the Ne\u{c}iporuk method for certain function
problems.
\end{itemize}

The next major step is to prove good lower bounds for trees of
height $h=4$.  If we can prove the above Thrifty Hypothesis for
deterministic BPs solving the function problem (and hence the
decision problem) for trees of height 4, then we would beat
the $\Omega(n^2)$ limitation mentioned above on Ne\u{c}iporuk's method.
See Section \ref{s:conclu} (Conclusion) for this argument, and
a comment about the nondeterministic case.

\subsection{Relation to previous work}
Taitslin \cite{ta05} proposed a problem similar to \bttwoh\
in which the functions attached to internal nodes are specific
quasi groups, in an unsuccessful attempt to prove $\nl\ne\p$.

Gal, Koucky and McKenzie \cite{gakomc08} proved exponential
lower bounds on the size of restricted $n$-way branching
programs solving versions of the problem GEN.  Like our problems
\bt\ and \ft, the best known upper bounds for solving GEN
come from pebbling algorithms.

As a concrete approach to separating \ncone\ from \nctwo, Karchmer,
Raz and Wigderson \cite{karawi95} suggested proving that the circuit
depth required to compose a Boolean function with itself $h$ times
grows appreciably with $h$. They proposed the \emph{universal
  composition relation} conjecture, stating that an abstraction of the
composition problem requires high communication complexity, as an
intermediate goal to validate their approach.  This conjecture was
later proved in two ways, first \cite{edimrusg01} using innovative
information-theoretic machinery and then \cite{hawi93} using a clever
new complexity measure that generalizes the subadditivity property
implicit in Ne\u{c}iporuk's lower bound method \cite{ne66}.  Proving
the conjecture thus cleared the road for the approach, yet no
sufficiently strong unrestricted circuit lower bounds could be proved
using it so far.

Edmonds, Impagliazzo, Rudich and Sgall \cite{edimrusg01} noted that the
approach would in fact separate \ncone\ from \acone. They also coined the
name \emph{Iterated Multiplexor} for the most general computational
problem considered in \cite{karawi95}, namely
composing in a tree-like fashion a set of
explicitly presented Boolean functions, one per tree node.
Our problem \ft\ can be considered as a generalization of
the Iterated Multiplexor problem in which the functions map
$[k]^d$ to $[k]$ instead of $\{0,1\}^d$ to $\{0,1\}$.
This generalization allows us to focus on getting
lower bounds as a function of $k$ when the tree is fixed.

For time-restricted branching programs,
Borodin, Razborov and Smolensky \cite{borasm93}
exhibited a family of Boolean functions that require exponential size
to be computed by nondeterministic syntactic read-$k$ times BPs.
Later Beame, Saks, Sun, and Vee \cite{BSSV03} exhibited such functions
that require exponential size to be computed by randomized
BPs whose computation
time is limited to $o(n\sqrt{\log n/\log\log n})$,
where $n$ is the input length.  However all these functions can
be computed by polynomial size BPs when time is unrestricted.

In the present paper we consider branching programs with no
time restriction such as read-$k$ times.  However the smallest
size deterministic BPs known to us that solve \ft\ implement
the black pebbling algorithm, and these BPs happen to be
(syntactic) read-once.

\subsection{Organization}
The paper is organized as follows.
Section \ref{s:preliminaries} defines the main notions used
in this paper, including branching programs and pebbling.
Section \ref{s:Connecting} relates pebbling and branching programs to
Turing machine space, noting in particular that a $k$-way BP size
lower bound of $\Omega(k^{\mbox{\scriptsize function}(h)})$ for \bt\ would
show $\lspace \neq \logcfl$.
Section \ref{s:PebBounds} proves upper and lower bounds on the
number of pebbles required to black, black-white and fractionally
pebble the tree $T^h_d$.
These pebbling bounds are exploited in Section \ref{s:PBbounds} to prove
upper bounds on the size of branching programs.
BP lower bounds are obtained using the Ne\u{c}iporuk method in Subsection
\ref{s:NecLB}.  Alternative proofs to some of these lower bounds
using the ``state sequence method'' are given in Subsection
\ref{s:beating}.
An example of a function problem for which the state sequence method
beats the Ne\u{c}iporuk method is given in Theorems
\ref{t:rootfunction} and \ref{t:childLB}.
Subsection \ref{s:thriftyLB} contains bounds for thrifty branching
programs.

\section{Preliminaries}\label{s:preliminaries}

We assume some familiarity with complexity theory, such as can be
found in \cite{go08}.
We write $[k]$ for $\{1,2,\ldots,k\}$.
For $d,h\ge 2$ we use $T_d^h$ to denote the balanced $d$-ary tree
of height $h$.

\medskip

\noindent
{\bf Warning:}  Here the {\em height} of a tree is the number
of levels in the tree, as opposed to the distance from root to
leaf.  Thus $T^2_2$ has just 3 nodes.

\medskip

\noindent
We number the nodes of $T_d^h$ as suggested by the heap data
structure.  Thus the root is node 1, and in general the children
of node $i$ are (when $d=2$) nodes $2i,2i+1$
(see Figure \ref{sample}).


\begin{defn}[Tree evaluation problems]
\label{d:treeEval} \
Given: The tree $T_d^h$ with each non-leaf node $i$
independently labeled with a function $f_i: [k]^d\rightarrow [k]$ and each
leaf node independently labeled with an element from $[k]$,
where $d,h,k\geq 2$.

\noindent
\emph{Function evaluation problem} \ft: Compute
the value $v_1\in[k]$ of the root $1$ of $T_d^h$, where in general
$v_i=a$ if $i$ is a leaf labeled
$a$ and $v_i=f_i(v_{j_1},\ldots,v_{j_d})$ if
the children of $i$ are $j_1,\ldots,j_d$.

\noindent
\emph{Boolean problem} \bt: Decide whether $v_1=1$.
\end{defn}

\subsection{Branching programs}

A family of branching programs serves as a nonuniform model of
of a Turing machine.  For each input size $n$ there is a BP
$B_n$ in the family which models the machine on inputs of size $n$. 
The states (or nodes) of $B_n$ correspond
to the possible configurations of the machine for inputs of size $n$.
Thus if the machine computes in space $s(n)$ then $B_n$ has
$2^{O(s(n))}$ states.

Many variants of the branching program model have been studied (see in
particular the survey by Razborov \cite{ra91} and the book by
Ingo Wegener \cite{we00}).  Our definition below is inspired by
Wegener \cite[p.\ 239]{we00}, by the $k$-way branching
program of Borodin and Cook \cite{boco82} and by its nondeterministic
variant \cite{borasm93,gakomc08}.  We depart from the latter however
in two ways: nondeterministic branching program labels are attached to
states rather than edges (because we think of branching program states
as Turing machine configurations) and cycles in branching programs are
allowed (because our lower bounds apply to this more powerful model).

\begin{defn}[Branching programs]
A \emph{nondeterministic $k$-way branching program} $B$ computing a
  total function $g:[k]^m\rightarrow R$, where $R$ is a finite set, is
  a directed rooted multi-graph whose nodes are called {\em states}. Every
  edge has a label from $[k]$.  Every state has a label from $[m]$,
  except $|R|$ {\em final} sink states consecutively labelled with the
  elements from $R$.  An input
  $(x_1,\ldots,x_m)\in [k]^m$ activates, for each $1\leq j\leq m$,
  every edge labelled $x_j$ out of every state labelled $j$. A
  {\em computation} on input $\vec{x}=(x_1,\ldots,x_m)\in [k]^m$
  is a directed path consisting of edges activated by $\vec{x}$
  which begins with the
  unique start state (the root), and either it is infinite,
  or it ends in the final state labelled
  $g(x_1,\ldots,x_m)$, or it ends in a nonfinal state labelled $j$
  with no outedge labelled $x_j$ (in which case we say the
  computation {\em aborts}).
  At least one such computation must end in a final state.
  The \emph{size} of $B$ is its number of states.  $B$ is
  \emph{deterministic $k$-way} if every non-final state has precisely
  $k$ outedges labelled $1,\ldots,k$.  $B$ is \emph{binary} if $k=2$.

We say that $B$ solves a decision problem (relation) if it computes
the characteristic function of the relation.
\end{defn}

A $k$-way branching program computing the function \ft\ requires
$k^d$ $k$-ary arguments for each internal node $i$ of $T^h_d$ in
order to specify the function $f_i$, together with one $k$-ary
argument for each leaf.  Thus in the notation of
Definition~\ref{d:treeEval}, \ft $: [k]^m \rightarrow R$ where $R=[k]$ and
$m=\frac{d^{h-1}-1}{d-1}\cdot k^d +  d^{h-1}$.
Also \bt $: [k]^m \rightarrow \{0,1\}$.

For fixed $d,h$ we are interested in how the number of states
required for a $k$-way branching program to compute \ft\ and \bt\
grows with $k$.  We define \dFstate\ (resp. \ndFstate)
to be the minimum number of states required for a deterministic
(resp. nondeterministic) $k$-way branching program to solve \ft.
Similarly we define \dBstate\ and \ndBstate\
to be the number of states for solving \bt.

The next lemma shows that the function problem is not much harder
to solve than the Boolean problem.

\begin{lem}\label{l:FvsB}
\begin{align*}
\dBstate \le \dFstate \le  k \cdot \dBstate\\
\ndBstate \le \ndFstate \le  k \cdot \ndBstate
\end{align*}
\end{lem}

\begin{proof}
The left inequalities are obvious.  For the others, we can
construct a branching program solving the function problem
from a sequence of $k$ programs solving Boolean problems, where
the $i$th program determines whether the value of the root node
is $i$.
\end{proof}

Next we introduce thrifty programs, a restricted form of $k$-way branching
programs for solving tree evaluation problems.  Thrifty
programs efficiently simulate pebbling algorithms, and
implement the best known upper bounds for
\ndBstate\ and \dFstate, and are within a factor of $\log k$
of the best known upper bounds for \dBstate.
In Section \ref{s:PBbounds} we prove tight lower bounds for deterministic
thrifty programs which solve \bt\ and \ft.

\begin{defn}[Thrifty branching program]
\label{d:thrifty}
A deterministic $k$-way branching program which solves \ft\
or \bt\ is {\em thrifty} if during the computation on any input every
query $f_i(\vec{x})$ to an internal node $i$ of $T^h_d$ satisfies the
condition that
$\vec{x}$ is the tuple of correct values for the children of node $i$.
A nondeterministic such program is {\em thrifty} if for every input
every computation which ends in a final state
satisfies the above restriction on queries.
\end{defn}

Note that the restriction in the above definition is semantic,
rather than syntactic.  It somewhat resembles the semantic
restriction used to define incremental branching programs
in \cite{gakomc08}.  However we are able to prove strong lower bounds
using our semantic restriction, but in \cite{gakomc08} a syntactic
restriction was needed to prove lower bounds.

\subsection{One function is enough}\label{s:oneFunction}

The theorem in this section is not used in the sequel.

It turns out that the complexities of \ft\ and \bt\ are not much different
if we require all functions assigned to internal nodes to be
the same.\footnote{We thank Yann Strozecki, who posed this question} 
To denote this restricted version of the problems
we replace $F$ by $\hat{F}$
and $B$ by $\hat{B}$.  Thus $\hat{F}T_d^h(k)$ is the function
problem for $T^h_d$ when all node functions are the same, and
$\hat{B}T_d^h(k)$ is the corresponding Boolean problem.
To specify an instance of one of these new problems we need only give
one copy of the table for the common node function $\hat{f}$,
together with the values for the leaves.

\begin{theorem}\label{t:single}
Let $N = (d^h-1)/(d-1)$ be the number of nodes in the tree $T^h_d$.
Any $Nk$-way branching program $\hat{B}$ solving $\hat{F}T_d^h(Nk)$
(resp. $\hat{B}T_d^h(Nk)$) can
be transformed to a $k$-way branching program $B$ solving
\ft\ (resp. \bt), where $B$ has no more states than $\hat{B}$ and
$B$ is deterministic iff $\hat{B}$ is deterministic. Also
for each $d\ge 2$ the decision problem $BT_d(h,k)$ is log space
reducible to $\hat{B}T_d(h,k)$ (where $h,k$ are input parameters).
\end{theorem}

\begin{proof}
Given an instance $I$ of \ft\ (or \bt) we can find a corresponding
instance $\hat{I}$ of $\hat{F}T_d^h(Nk)$ (or $\hat{B}T_d^h(Nk)$) by
coding the set
of all functions $f_i$ associated with internal nodes $i$
in $I$ by a single function $\hat{f}$ associated with each node
of $\hat{I}$. Here we represent each element
of $[Nk]$ by a pair $\langle i,x\rangle$, where $i\in [N]$
represents a node in $T^h_d$ and $x\in [k]$.
We want to satisfy the following Claim:

{\bf Claim:} If a node $i$ has a value $x$ in $I$ then node $i$ has
value $\langle i,x\rangle$ in $\hat{I}$.

Thus if $i$ is a leaf node, then we define the leaf value for node $i$
in $\hat{I}$ to be $\langle i,x\rangle$, where $x$ is the
value of leaf $i$ in $I$.

We define the common internal node function $\hat{f}$ as follows. 
If nodes $i_1,\ldots,i_d$ are the children of node $j$ in $T^h_d$, then
\begin{equation}\label{e:fhat}
  \hat{f}(\langle i_1,x_1\rangle, \ldots,\langle i_d,x_d\rangle)=
                \langle j,f_j(x_1,\ldots,x_d)\rangle
\end{equation}
The value of $\hat{f}$ is irrelevant (make it $\langle 1,1\rangle$)
if  nodes $i_1,\ldots,i_d$ are not the children of $j$.

An easy induction on the height of a node $i$ shows that the
above {\bf Claim} is satisfied.

Note that the value $x$ of the root node $1$ in $I$ is easily determined
by the value $\langle 1,x\rangle$ of the root in $\hat{I}$.
We specify that the pair $\langle 1,1\rangle$ has value 1 in $[N]$,
so $I$ is a YES instance of the decision problem \bt\ iff $\hat{I}$
is a YES instance of $\hat{B}T_d^h(Nk)$.

To complete the proof of the last sentence in the theorem we note
that the number of bits needed to specify $I$ is
$\Theta(Nk^d\log k)$, and the number of bits to specify $\hat{I}$
is dominated by the number to specify $\hat{f}$, which is
$O((Nk)^d\log(Nk))$.  Thus the transformation from $I$ to $\hat{I}$
is length-bounded by a polynomial in length of its argument,
and it is not hard to see that it can be carried out in log space.

Now we prove the first part of the theorem.
Given an $Nk$-way BP $\hat{B}$ solving $\hat{B}T_d^h(Nk)$
(resp. $\hat{F}T_d^h(Nk)$) we can find a corresponding
$k$-way BP $B$ solving \bt\ (resp. \ft) as follows.

The idea is that on input instance $I$,  $B$ acts like $\hat{B}$
on input $\hat{I}$.   Thus for each state $\hat{q}$ in $\hat{B}$
that queries a leaf node $i$, the corresponding state $q$ in
$B$ queries $i$, and for each possible answer $x\in [k]$, $B$ has an
outedge labelled $x$ corresponding to the edge from $\hat{q}$ labelled
$\langle i,x\rangle$.  If $\hat{q}$ queries $\hat{f}$ at arguments
as in (\ref{e:fhat}) (where $i_1,\ldots,i_d$ are the children of
node $j$) then $q$ queries $f_j(x_1,\ldots,x_d)$ and for each
$x\in [k]$, $q$ has an outedge labelled $x$ corresponding to the
edge from $\hat{q}$
labelled $\langle j,x\rangle$.  If $i_1,\ldots,i_d$ are not the
children of $j$, then the node $q$ is not necessary in $B$,
since the answer to the query is always the default $\langle 1,1\rangle$.

In case $\hat{B}$ is solving the function problem $\hat{F}T_d^h(Nk)$
then each output state labelled $\langle 1,x\rangle$ is relabelled
$x$ in $B$ (recall that the root of $T^h_d$ is number 1).
Any output state $q$ labelled $\langle i,x\rangle$ where $i>1$ will
never be reached in $B$ (since the value of the root node
of $\hat{I}$ always has the form $\langle 1,x\rangle$) so
$q$ can be deleted.
For any edge in $\hat{B}$ leading to $q$ the
corresponding edge in $B$ can lead anywhere.
\end{proof}

One goal of this paper is to motivate trying to show
$BT_d(h,k) \notin \lspace$.
By Theorem \ref{t:single} this is equivalent to showing
$\hat{B}T_d(h,k)\notin \lspace$.  Further our suggested method
is to try proving for each fixed $h$ a lower bound of $\Omega(k^{r(h})$
on the number of states required for a $k$-way BP to solve \ft,
where $r(h)$ is any unbounded function
(see Corollary \ref{c:thegoal} below).  Again acording to
Theorem \ref{t:single} (since $N$ is a constant)  technically
speaking we may as well
assume that all the node functions
in the instance of \ft\ are the same.  However in practice this
assumption is not helpful in proving a lower bound.
For example Theorem \ref{t:beatittwice} states that $k^3$ states
are required for a deterministic $k$-way BP to solve $FT^3_2(k)$,
and the proof assigns three different functions to the three
internal nodes of the binary tree of height 3.


\subsection{Pebbling}

The pebbling game for dags was defined by Paterson and Hewitt
\cite{pahe70} and was used as an abstraction for deterministic
Turing machine space in \cite{co74}.  Black-white pebbling
was introduced in \cite{cose76} as an abstraction of
nondeterministic Turing machine space (see \cite{nordstrom}
for a recent survey).

Here we define and use three versions of the pebbling game.
The first is a simple `black pebbling' game:  A black
pebble can be placed on any leaf node, and in general
if all children of a node $i$ have pebbles, then one of the
pebbles on the children can be slid to $i$ (this is a
``black sliding move')'.  Any black pebble can be removed
at any time.  The goal is to pebble the root, using as few
pebbles as possible.  The second
version is `whole' black-white pebbling as defined in
\cite{cose76} with the restriction that we do not allow
``white sliding moves''.  Thus if node $i$ has a white
pebble and each child of $i$ has a pebble (either black or white)
then the white pebble can be removed.  (A white sliding move
would apply if one of the children had no pebble, and the
white pebble on $i$ was slid to the empty child.  We do not
allow this.)  A white pebble can be placed on any node at
any time.  The goal is to start and end with no pebbles,
but to have a black pebble on the root at some time.

The third is a new game called {\em fractional pebbling},
which generalizes whole black-white pebbling by allowing the
black and white pebble value of a node to be any real number
between 0 and 1.  However the total pebble value of each
child of a node $i$ must be 1 before the black value of $i$
is increased or the white value of $i$ is decreased.
Figure \ref{f:bin_h3_fract_ub} illustrates two configurations
in an optimal fractional pebbling of the binary tree of
height three using 2.5 pebbles.

Our motivation for choosing these definitions is that we want
pebbling algorithms for trees to closely correspond to $k$-way
branching program algorithms for the tree evaluation problem.

We start by defining fractional pebbling, and then define
the other two notions as restrictions on fractional pebbling.

\begin{defn}[Pebbling]
\label{d:pebbling}
A {\em fractional pebble configuration} on a rooted $d$-ary
tree $T$ is an assignment of
a pair of real numbers $(b(i),w(i))$ to each node $i$ of the tree, where
\begin{align}
   &  0\le b(i),w(i) \label{e:consOne} \\ &  b(i)+w(i)\le 1
   \label{e:consTwo}
\end{align}
 Here $b(i)$ and $w(i)$ are the
{\em black pebble value} and the {\em white pebble value}, respectively,
of $i$, and $b(i)+w(i)$ is the {\em pebble value} of $i$.  The number of
pebbles in the configuration is the sum over all nodes $i$ of the pebble
value of $i$.  The legal pebble moves are as follows (always subject to
maintaining the constraints (\ref{e:consOne}), (\ref{e:consTwo})): (i)
For any node $i$, decrease $b(i)$ arbitrarily, (ii) For any node $i$,
increase $w(i)$ arbitrarily, (iii) For every node $i$,
if each child of
$i$ has pebble value 1, then decrease $w(i)$ to 0, increase $b(i)$
arbitrarily, and simultaneously decrease the black pebble values of
the children of $i$ arbitrarily.

A {\em fractional pebbling} of $T$ using $p$ pebbles is any
sequence of (fractional) pebbling moves on nodes of $T$ which starts
and ends with every node having pebble value 0, and at some point the
root has black pebble value 1, and no configuration has more than $p$
pebbles.

A {\em whole black-white pebbling} of $T$ is a fractional pebbling of
$T$ such that $b(i)$ and $w(i)$ take values in $\{0,1\}$ for every node
$i$ and every configuration.  A {\em black pebbling} is a
black-white pebbling in which $w(i)$ is always 0.
\end{defn}

Notice that rule (iii) does not quite treat black and white pebbles
dually, since the pebble values of the children must each be 1 before any
decrease of $w(i)$ is allowed.  A true dual move would allow increasing
the white pebble values of the children so they all have pebble value 1
while simultaneously decreasing $w(i)$.  In other words, we allow black
sliding moves, but disallow white sliding moves.  The reason for this
(as mentioned above) is
that nondeterministic branching programs can simulate the former,
but not the latter.

We use $\Bpebbles(T)$, $\BWpebbles(T)$, and $\FRpebbles(T)$ respectively
to denote the minimum number of pebbles required to black pebble $T$,
black-white pebble $T$, and fractional pebble $T$.
Bounds for these values are given in Section \ref{s:PebBounds}.
For example for $d=2$ we have $\Bpebbles(T^h_2)= h$,
$\BWpebbles(T^h_2)= \lceil h/2\rceil +1$, and
$\FRpebbles(T^h_2) \le h/2+1$.  In particular
$\FRpebbles(T^3_2) = 2.5$ (see Figure \ref{f:bin_h3_fract_ub}).

\section{Connecting TMs, BPs, and Pebbling}\label{s:Connecting}

Let \fthk\ be the same as \ft\ except now the inputs vary
with both $h$ and $k$, and we assume the input to \fthk\
is a binary string $X$ which codes $h$ and $k$ and
codes each node function $f_i$ for the tree $T^h_d$
by a sequence of $k^d$ binary numbers and
each leaf value by a binary number in $[k]$, so $X$ has length
\begin{equation}\label{e:Flength}
   |X| = \Theta(d^hk^d\log k)
\end{equation}
The output is a binary number in $[k]$ giving
the value of the root.

The problem \bthk\ is the Boolean version of \fthk:
The input is the same, and the instance is true iff the value of the
root is 1.

Obviously \bthk\ and \fthk\ can be solved in polynomial time,
but we can prove a stronger result.

\begin{theorem}\label{t:logDCFL}
The problem \bthk\ is in \logdcfl, even when $d$ is given as an
input parameter.
\end{theorem}

\begin{proof}
By \cite{su78} if suffices to show that \bthk\ is solved
by some deterministic auxiliary pushdown automaton $M$ in $\log$ space
and polynomial time.  The algorithm for $M$ is to use its stack
to perform a depth-first search of the tree $T^h_d$, where for
each node $i$ it keeps a partial list of the values of the
children of $i$, until
it obtains all $d$ values, at which point it computes the value of
$i$ and pops its stack, adding that value to the list for
the parent node.

Note that the length $n$ of an input instance is about $d^k k^d\log k$
bits, so $\log n > d\log k$, so $M$ has ample space on its 
work tape to write all $d$ values of the children of a node $i$.
\end{proof}

The best known upper bounds on branching program size for \ft\
grow as $k^{\Omega(h)}$.
The next result shows (Corollary~\ref{c:thegoal}) that any lower bound
with a nontrivial dependency on $h$ in the exponent of $k$ for
deterministic (resp. nondeterministic) BP size would
separate \lspace (resp. \nl) from \logdcfl.

\begin{theorem}\label{t:goal}
For each $d \ge 2$, if \bthk\ is in \lspace\ (resp.\ \nl)
then there is a constant $c_d$ and a function $f_d(h)$
such that $\dFstate \le f_d(h)k^{c_d}$
(resp. $\ndFstate \le f_d(h)k^{c_d}$) for all $h,k\ge 2$.
\end{theorem}

\begin{proof}
By Lemma \ref{l:FvsB} it suffices to prove this for \dBstate\
and \ndBstate\ instead of \dFstate\ and \ndFstate.
In general a Turing machine which can enter at most $C$ different
configurations on all inputs of a given length $n$ can be
simulated (for inputs of length $n$) by a binary (and hence $k$-ary)
branching program with $C$ states.
Each Turing machine using space $O(\log n)$ has at most $n^c$
possible configurations on any input of length $n \ge 2$,
for some constant $c$.  By (\ref{e:Flength}) the input for
\bthk\ has length $n=\Theta(d^hk^d\log k)$, so there are
at most $(d^hk^d\log k)^{c'}$ possible configurations for a log
space Turing machine solving \bthk, for some constant $c'$.  So we can
take $f_d(h) = d^{c'h}$ and $c_d = c'(d+1)$.
\end{proof}

\begin{cor}\label{c:thegoal}
Fix $d \ge 2$ and any unbounded function $r(h)$. If
$\dFstate\in\Omega(k^{r(h)})$ then $\bthk\notin\lspace$.
If $\ndFstate\in\Omega(k^{r(h)})$
then $\bthk\notin\nl$.
\end{cor}

The next result connects pebbling upper bounds with upper
bounds for thrifty branching programs.

\begin{theorem}\label{t:pebSim}
(i) If $T^h_d$ can be black pebbled with $p$ pebbles, then deterministic
thrifty branching programs with $O(k^p)$ states can solve
\ft\ and \bt. 

(ii) If $T^h_d$ can be fractionally pebbled with $p$
pebbles then nondeterministic thrifty branching programs can
solve \bt\ with $O(k^p)$ states.
\end{theorem}

\begin{proof}
Consider the sequence $C_0,C_1,\ldots C_\tau$ of pebble configurations
for a black pebbling of $T^h_d$ using $p$ pebbles.  We may as
well assume that the root is pebbled in configuration $C_\tau$,
since all pebbles could be removed in one more step at no
extra cost in pebbles.  We design
a thrifty branching program $B$ for solving \ft\ as follows.
For each pebble configuration $C_t$, program $B$ has $k^p$ states; one
state for each possible assignment of a value from $[k]$ to each
of the $p$ pebbles.  Hence $B$ has $O(k^p)$ states, since $\tau$
is a constant independent of $k$.  Consider an input $I$ to \ft, and let
$v_i$ be the value in $[k]$ which $I$ assigns to node $i$ in
$T^h_d$ (see Definition \ref{d:treeEval}).  We design $B$ so that on
$I$ the computation of $B$ will be a state sequence
$\alpha_0,\alpha_1,\ldots,\alpha_\tau$, where the state $\alpha_t$
assigns to each pebble the value $v_i$ of the node $i$ that it is on.
(If a pebble is not on any node, then its value is 1.)

For the initial pebble configuration no pebbles have been assigned
to nodes, so the initial state of $B$ assigns the value 1 to each
pebble.  In general if $B$ is in a state $\alpha$ corresponding
to configuration $C_t$, and the next configuration $C_{t+1}$
places a pebble $j$ on node $i$, then the state $\alpha$
queries the node $i$ to determine $v_i$, and moves to a new
state which assigns $v_i$ to the pebble $j$
and assigns 1 to any pebble which is removed from the
tree.  Note that if $i$ is an internal node, then all children
of $i$ must be pebbled at $C_t$, so the state $\alpha$
`knows' the values $v_{j_1},\ldots,v_{j_d}$ of the children of
$i$, so $\alpha$ queries $f_i(v_{j_1},\ldots,v_{j_d})$.

When the computation of $B$ reaches a state $\alpha_\tau$ corresponding
to $C_\tau$, then $\alpha_\tau$ determines the value of the root
(since $C_\tau$ has a pebble on the root), so $B$ moves to
a final state corresponding to the value of the root.

The argument for the case of whole black-white pebbling is
similar, except now the value for each white pebble represents
a guess for the value $v_i$ of the node it is on.  If
the pebbling algorithm places a white pebble $j$ on a node at some
step, then the corresponding state of $B$ nondeterministically
moves to any state in which the values of all pebbles except $j$
are the same as before, but the value of $j$ can be any value
in $[k]$.  If the pebbling algorithm removes a white pebble $j$
from a node $i$, then the corresponding state has a guess $v'_i$
for the value of $i$, and either $i$ is a leaf, or all children
of $i$ must be pebbled.  The corresponding state of $B$
queries $i$ to determine its true value $v_i$.  If $v_i \ne v'_i$
then the computation aborts (i.e. all outedges from the state have
label $v'_i$).  Otherwise $B$ assigns $j$ the value 1 and continues.

When $B$ reaches a state $\alpha$ corresponding to a pebble configuration
$C_t$ for which the root has a black pebble $j$, then $\alpha$
knows whether or not the tentative value assigned to the root is 1.
All future states remember whether the tentative value is 1.
If the computation successfully (without aborting) reaches a state
$\alpha_\tau$ corresponding to the final pebble configuration
$C_\tau$, then $B$ moves to the final state corresponding to
output 1 or output 0, depending on whether the tentative root
value is 1.

Now we consider the case in which $C_0,\ldots,C_\tau$ represents
a fractional pebbling computation.   If $b(i),w(i)$ are the black
and white pebbled values of node $i$ in configuration $C_t$, then
a state $\alpha$ of $B$ corresponding to $C_t$ will remember
a fraction $b(i) + w(i)$ of the $\log k$ bits specifying the value
$v_i$ of the node $i$, where the fraction $b(i)$ of bits are verified,
and the fraction $w(i)$ of bits are conjectured.
In general these numbers
of bits are not integers, so they are rounded up to the next integer.
This rounding introduces at most two extra bits for each node in $T^h_d$,
for a total of at most $2T$ extra bits, where $T$ is the number of
nodes in $T^h_d$.  Since the sum over all nodes of all pebble
values is at
most $p$, the total number of bits that need to be remembered for a
given pebble configuration is at most $p \log k + 2T$, where $T$
is a constant.  Associated with each step in the fractional pebbling
there are $2^{p\log k +2T} = O(k^p)$ states in the branching program,
one for each setting of these bits.   These bits can be updated
for each of the three possible fractional pebbling moves
(i), (ii), (iii) in Definition \ref{d:pebbling} in a manner similar
to that for whole black-white pebbling.

It is easy to see that in all cases the branching programs
described satisfy the thrifty requirement that an internal node
is queried only at the correct values for its children
(or, in the black-white and fractional cases, the program
aborts if an incorrect query is made because of an incorrect
guess for the value of a white-pebbled node).
\end{proof}

\begin{cor}
$\dFstate = O(k^{\Bpebbles(T^h_d)})$ and
$\ndFstate = O(k^{\FRpebbles(T^h_d)})$.
\end{cor}

\section{Pebbling Bounds}\label{s:PebBounds}

\subsection{Previous results}

We start by summarizing what is known about whole black and black-white
pebbling numbers as defined at the end of Definition \ref{d:pebbling}
(i.e. we allow black sliding moves but not white sliding moves).

The following are minor adaptations of results and techniques that have
been known since work of Loui, Meyer auf der Heide and Lengauer-Tarjan
\cite{loui, meyeraufderheide,lengauer-tarjan} in the late '70s. They
considered pebbling games where sliding moves were either disallowed or
permitted for both black and white pebbles, in contrast to our results
below.

We always assume $h\ge 2$ and $d\ge 2$.

\begin{theorem} \label{t:blackSliding}
$\Bpebbles(T^{h}_{d}) = (d-1)h - d +2$.
\end{theorem}

\begin{proof} For $h=2$ this gives $\Bpebbles(T^{2}_{d}) = d$, which is
obviously correct.  In general we
show $\Bpebbles(T^{h+1}_{d}) = \Bpebbles(T^h_{d}) + d - 1$, from which
the theorem follows.

The following pebbling strategy gives the upper bound: Let the root be
node $1$ and the children be $2 \ldots d+1$. Pebble the nodes
$2 \ldots d+1$
in order using the optimal number of pebbles for $T^{h-1}_{d}$,
leaving a black pebble at each node. Note that for the black
pebble game, the complexity of pebbling in the game where a
pebble remains on the root is the
same as for the game where the root has a black pebble on it at
some point.  The maximum number of pebbles at any point on the tree is
$d-1 + \Bpebbles(T^{h-1}_{d})$. Now slide the black pebble from node $1$ to
the root, and then remove all pebbles.

For the lower bound, consider the time $t$ at which the children of the
root all have black pebbles on them. There must be a final time $t'$
before $t$ at which one of the sub-trees rooted at $2,3, \ldots d+1$ had
$T^{h}_{d}$ pebbles on it. This is because pebbling any of these
subtrees requires at least $T^{h}_{d}$ pebbles, by definition.
At time $t'$, all the other subtrees must
have at least 1 black pebble each on them. If not, then there is a subtree
$T$ which does not, and it would have to be pebbled before time $t$, which
contradicts the definition of $t'$. Thus at time $t'$, there are at least
$T^{h}_{d} + d - 1$ pebbles on the tree.
\end{proof}

\begin{theorem}\label{t:BSlideW}
For $d=2$ and $d$ odd:\,
\begin{equation}\label{e:dOdd}
\BWpebbles(T^{h}_{d}) = \lceil(d-1)h/2\rceil +1
\end{equation}
For $d$ even:\, 
\begin{equation}\label{e:dEven}
\BWpebbles(T^{h}_{d}) \leq \lceil(d-1)h/2\rceil +1
\end{equation}
When $d$ is odd, this number is the same as when white sliding
moves are allowed.
\end{theorem}

\begin{proof}
We divide the proof into three parts.

\medskip

\noindent
{\bf Part I:}\\
We show (\ref{e:dOdd}) when $d$ is odd.

For $h=2$ this gives $\BWpebbles(T^{2}_{d}) = d$, which is
obviously correct.  In general for odd $d$ we
show
\begin{equation}\label{e:BSWrec}
\BWpebbles(T^{h+1}_{d}) = \BWpebbles(T^h_{d}) + (d - 1)/2
\end{equation}
from which the theorem follows for this case.

For the upper bound for the left hand side,
we strengthen the induction hypothesis by
asserting that during the pebbling there is a {\em critical time}
at which the root has a black pebble and there are at most
$\BWpebbles(T^h_d)-(d-1)/2$ pebbles on the tree (counting the pebble
on the root).  This can be made true when $h=2$ by removing all
the pebbles on the leaves after the root is pebbled.

To pebble the tree $T^{h+1}_d$, note that we are allowed $(d-1)/2$
extra pebbles over those required to pebble $T^h_d$.  Start
by placing black pebbles on the left-most $(d-1)/2$ children
of the root, and removing all other pebbles.  Now go through
the procedure for pebbling the middle principal subtree, stopping at the
critical time, so that there is a black pebble on the middle
child of the root and at most $\BWpebbles(T^h_d)-(d-1)/2$ pebbles on
the middle subtree.  Now place white pebbles on the remaining
$(d-1)/2$ children of the root, slide a black pebble to the root,
and remove all black pebbles on the children of the root.  This is
the critical time for pebbling $T^{h+1}_d$: note that there are
at most $\BWpebbles(T^h_d)$ pebbles on the tree (we removed the black
pebble on the root of the middle subtree).

Now remove the pebble on the root and remove all pebbles on
the middle  subtree by completing its pebbling
(keeping the $(d-1)/2$ white pebbles on the children in place).
Finally remove the remaining $(d-1)/2$ white pebbles one by one,
simply by pebbling each subtree, and removing the white pebble
at the root of the subtree instead of black-pebbling it.

To prove the lower bound for the left hand side of (\ref{e:BSWrec}),
we strengthen the induction hypothesis so that now a black-white
pebbling allows white sliding moves, and the root may be pebbled
by either a black pebble or a white pebble.  (Note that for the
base case the tree $T^2_d$ still requires $d$ pebbles.)
Consider such a pebbling of $T^{h+1}_d$
which uses as few moves as possible.  Consider a time $t$
at which all children of the root have pebbles on them (i.e.
just before the root is black pebbled or just after a white
pebble on the root is removed).
For each child $i$, let $t_i$ be a time at which the tree rooted
at $i$ has $\BWpebbles(T^h_d)$ pebbles on it.  We may assume
$$
   t_2<t_3< \ldots < t_{d+1}
$$
Let $m = (d+3)/2$ be the middle child.  If $t_m < t$
then each of the $(d-1)/2$ subtrees rooted at $i$ for $i<m$
has at least one pebble on it at time $t_m$, since otherwise
the effort made to place $\BWpebbles(T^h_d)$ pebbles on it earlier
is wasted.
Hence (\ref{e:BSWrec}) holds for this case.  Similarly if
$t_m > t$ then each of the $(d-1)/2$ subtrees rooted at $i$ for $i>m$
has at least one pebble on it at time $t_m$, since otherwise the
effort to place $T^h_d$ pebbles on it later is wasted, so
again (\ref{e:BSWrec}) holds.

\medskip

\noindent
{\bf Part II:}\\
We prove (\ref{e:dEven}) for even degree $d$:
\[\BWpebbles(T^{h}_{d}) \leq \lceil(d-1)h/2\rceil +1\]

For $h=2$ the formula gives $\BWpebbles(T^2_d) = d$, which is obviously
correct.  For $h=3$ the formula gives $\BWpebbles(T^3_d) = (3/2)d$, which
can be realized by black-pebbling $d/2+1$ of the root's children
and white-pebbling the rest.  In general it suffices to prove
the following recurrence:
\begin{equation}\label{e:BSWrecEven}
\BWpebbles(T^{h+2}_{d}) \le \BWpebbles(T^h_{d}) + d-1
\end{equation}
We strengthen the induction hypothesis by
asserting that during the pebbling of $T^{h}_d$ there is a
{\em critical time}
at which the root has a black pebble and there are at most
$\BWpebbles(T^h_d)-(d-1)$ pebbles on the tree (counting the pebble
on the root).  This is easy to see when $h=2$ and $h=3$.

We prove the recurrence as follows.
We want to pebble $T^{h+2}_d$ using $d-1$ more pebbles than
is required to pebble $T^h_d$.   Let us call the children
of the root $c_1,\ldots,c_d$.  We start by placing black pebbles
on $c_1,\ldots c_{d/2}$.  We illustrate how to do this by
showing how to place a black pebble on $c_{d/2}$ after
there are black pebbles on nodes $c_1,\ldots c_{d/2-1}$.
At this point we still have $d/2$ extra pebbles left among
the original $d-1$.  Let us assign the names $c'_1,\ldots c'_d$
to the children of $c_{d/2}$.  Use the $d/2$ extra pebbles
to put black pebbles on $c'_1,\ldots,c'_{d/2}$.  Now run
the procedure for pebbling the subtree rooted at $c'_{d/2+1}$
up to the critical time, so there is a black pebble on $c'_{d/2+1}$.
Now place white pebbles on the remaining $d/2-1$ children
of $c_{d/2}$, slide a black pebble up to $c_{d/2}$, remove
the remaining black pebbles on the children of $c_{d/2}$, and
complete the pebbling procedure for the subtree rooted
at $c'_{d/2+1}$, so that subtree has no pebbles.  Now remove
the white pebbles on the remaining $d/2-1$ children of $c_{d/2}$
using the remaining $d/2-1$ extra pebbles.

At this point there are black pebbles on nodes $c_1,\ldots,c_{d/2}$,
and no other pebbles on the tree.  We now place a black pebble
on $c_{d/2+1}$ as follows.  Let us assign the names
$c''_1,\ldots c''_d$ to the children of $c_{d/2+1}$.  Use the
remaining $d/2-1$ extra pebbles to place black pebbles on
$c''_1,\ldots,c''_{d/2-1}$.  Now run the pebble procedure on
the subtree rooted at $c''_{d/2}$ up to the critical time, so
$c''_{d/2}$ has a black pebble.  Now place white pebbles on
the remaining $d/2$ children of $c_{d/2+1}$, slide a black
pebble up to $c_{d/2+1}$, remove the remaining black pebbles
on the children of $c_{d/2+1}$, place white pebbles on
the remaining $d/2-1$ children of the root, slide a black
pebble up to the root, and remove the remaining black
pebbles from the children of the root.

This is now the critical time for the procedure pebbling
$T^{h+2}_d$.  There is a black pebble on the root, $d/2-1$ white pebbles
on the children of the root, $d/2$ white pebbles on the
children of  $c_{d/2+1}$, and at most $\BWpebbles(T^h_d) -d$
pebbles on the subtree rooted at $c''_{d/2}$ (we've removed
the black pebble on $c''_{d/2}$), making a total of at most
$\BWpebbles(T^h_d)$ pebbles on the tree.

Now remove the black pebble from the root and complete the pebble
procedure for the subtree rooted at $c''_{d/2}$ to remove
all pebbles from that subtree.  There remain $d/2-1$ white
pebbles on the children of the root and $d/2$ white pebbles
on the children of $c_{d/2+1}$, making a total of $d-1$ white
pebbles.  Now remove each of the white pebbles on the children
of $c_{d/2+1}$ by pebbling each of these subtrees in turn.
Finally we can remove each of the remaining $d/2-1$ white pebbles
on the children of the root by a process similar to the one
used to place $d/2$ black pebbles on the children of the root
at the beginning of the procedure (we now in effect have one
more pebble to work with).

\medskip

\noindent
{\bf Part III:}\\
Finally we give the lower bound for the case $d=2$:
\[\BWpebbles(T^{h}_{2}) \geq \lceil h/2 \rceil +1\]

Clearly 2 pebbles are required for the tree of height 2, and it is easy
to show that 3 pebbles are required for the height 3 tree.

In general it suffices to show that the binary tree $T$ of height $h+2$
requires at least one more pebble than the binary tree of height $h$.
Suppose otherwise, and consider a pebbling of $T$ that uses
the minimum number of pebbles required for the tree of height $h$,
and assume that the pebbling is as short as possible.
Let $t_1$ be a time when the root has a black pebble. For $i=3,4,5$
there must be a time $t_i$ when all the pebbles are on the subtree rooted
at node $i$.  This is because node $i$ must be pebbled at some
point, and if the pebble is white then right after the white
pebble is removed we could have placed a black pebble in its place
(since we do not allow white sliding moves). 

Suppose that $\{t_1,t_3,t_4,t_5\}$ are ordered such that
$$
   t_{i_1}<t_{i_2}<t_{i_3}<t_{i_4}
$$
Then $t_1$ cannot be either $t_{i_3}$ or $t_{i_4}$ since
otherwise at time $t_{i_2}$ there are no pebbles on the subtree
rooted at node $i_1$ and hence its earlier pebbling was wasted
(since the root has yet to be pebbled).
Similarly if $t_1$ is either $t_{i_1}$ or $t_{i_2}$ then
at time $t_{i_3}$ there are no pebbles on the subtree rooted
at $i_4$, and since the root has already been pebbled the later
pebbling of this subtree is wasted.
\end{proof}

\subsection{Results for fractional pebbling}

The concept of fractional pebbling is new.  Determining the
minimum number $p$ of pebbles required to fractionally pebble $T^h_d$ is
important since $O(k^p)$ is the best known upper bound on the
number of states required by a nondeterministic BP to solve
\ft\ (see Theorem \ref{t:pebSim}).
It turns out that proving fractional
pebbling lower bounds is much more difficult than proving
whole black-white pebbling lower bounds.  We are able to get
exact fractional pebbling numbers for the binary tree of height 4
and less, but the best general lower bound comes
from a nontrivial reduction
to a paper by Klawe \cite{klawe} which proves bounds for
the pyramid graph.  This bound is within $d/2+1$ pebbles of optimal
for degree $d$ trees (at most 2 pebbles from optimal for binary trees).

Our proof of the exact value of $\FRpebbles(T^4_2) = 3$
led us to conjecture that any nondeterministic BP computing $BT_2^4(k)$
requires $\Omega(k^3)$ states. In section \ref{s:PBbounds} we provide
evidence for that conjecture by proving that any nondeterministic
\emph{thrifty} BP requires $O(k^3)$ states. The lower bound for height 3
and any degree follows from the lower bound of
$\Omega(k^{\frac{3}{2}d-\frac{1}{2}})$ states for
nondeterministic branching programs computing $BT_d^3(k)$
(Corollary \ref{c:HtThree}).

\ifpdf
\else
\begin{figure}
\vspace*{-.5cm}
\hspace*{1.5cm}\includegraphics[scale=0.70]{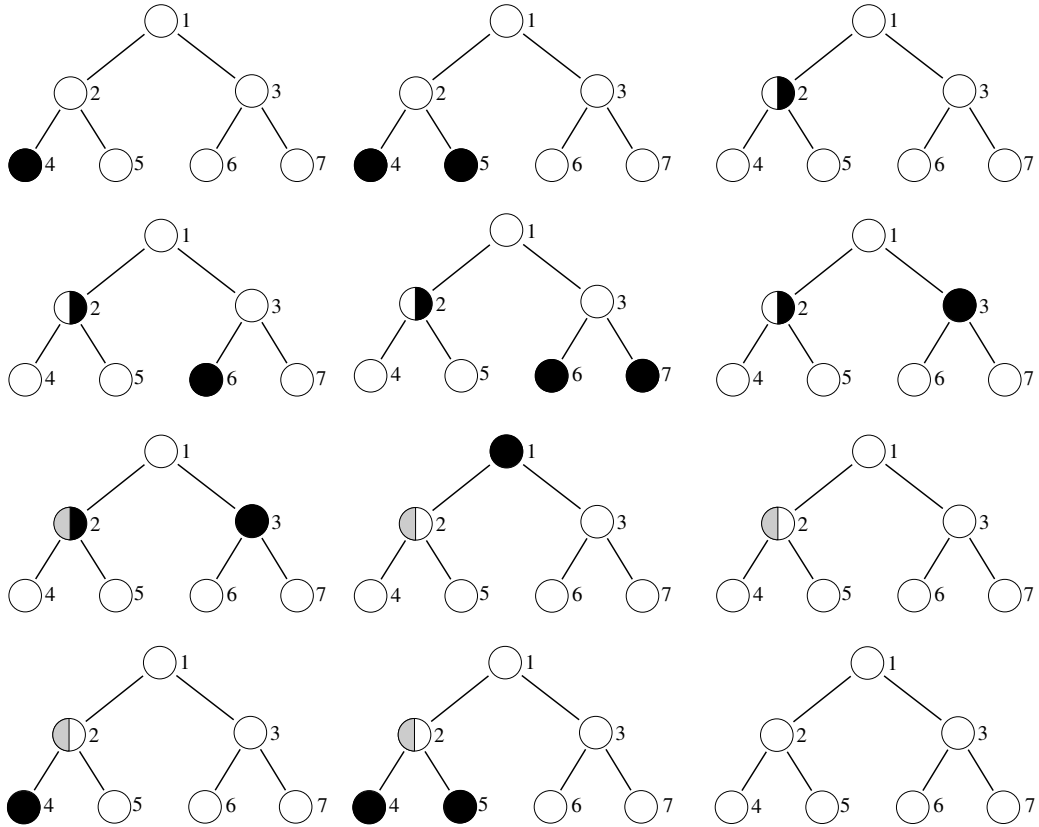}
\vspace*{-1cm}
\caption{An optimal fractional pebbling sequence for the height 3 tree
using 2.5 pebbles, all configurations included. The grey half circle
means the \emph{white} value of that node is $.5$, whereas unshaded
area means absence of pebble value. So for example in the seventh
configuration, node 2 has black value .5 and white value .5, node 3
has black value 1, and the remaining nodes all have black and white
value 0. }
\label{f:bin_h3_fract_ub}
\end{figure}
\fi

We start by
presenting a general result showing that fractional pebbling
can save at most a factor of two over whole black-white pebbling
for any DAG (directed acyclic graph).  (Here the pebbling rules
for a DAG are the same as for a tree, where we require that
every sink node (i.e. every `root') must have a whole black
pebble at some point.)   We will not use this result, but it
does provide a simple proof of weaker lower bounds than
those given in Theorem \ref{t:daryFract} below.

\begin{theorem}\label{r:factorTwo}
If a DAG $D$ has a fractional pebbling using $p$ pebbles,
then it has a black-white pebbling using $2p$ pebbles.
\end{theorem}

\begin{proof}
Given a sequence $P$ of fractional pebbling moves for a
DAG $D$ in which at most $p$ pebbles are used, we define a corresponding
sequence $P'$ of pebbling moves in which at most $2p$ pebbles are used.
The sequence $P'$ satisfies the following invariant with respect to $P$.

\medskip

\noindent ($\spadesuit$) A node $v$ has a black pebble (resp. a white
pebble) on it at time $t$ with respect to $P'$ iff $b(v) \geq 1/2$
(resp. $w(v) > 1/2$) at time $t$ with respect to $P$.

\medskip

An important consequence of this invariant is that if at time $t$ in
$P$ node $v$ satisfies $b(v)+w(v)=1$ then at time $t$ in $P'$ node $v$
is pebbled.

We describe when a pebble is placed or removed in $P$. At the
beginning, there are no pebbles on any nodes. $P'$ simulates $P$
as follows. Assume there is a certain configuration of pebbles on
$D$, placed according to $P'$ after time $t-1$; we describe how $P$'s
move at time $t$ is reflected in $P'$. If in the current move of $P$,
$b(v)$ (resp. $w(v)$) increases to $1/2$ or greater (resp. greater than
$1/2$) for some node $v$, then the current pebble, if any, on $v$, is
removed and a black pebble (resp. a white pebble) is placed on $v$ in
$P'$. Note that this is always consistent with the pebbling rules. If in
the current configuration of $P'$ there is a black (resp. white) pebble
on a vertex $v$, and in the current move of $P$, $b(v)$ (resp. $w(v)$)
falls below $1/2$, then the pebble on $v$ is removed. Again, this is
always consistent with the pebbling rules for the black-white pebble
game and the fractional black-white pebble game. For all other kinds of
moves of $P$, the configuration in $P'$ does not change.

If $P$ is a valid sequence of fractional pebbling moves, then $P'$ is a
valid sequence of pebbling moves. We argue that the cost of $P'$ is at
most twice the cost of $P$, and that if there is a point at which the root
has black pebble value $1$ with respect to $P$,
then there is a point at which
the root is black-pebbled in $P'$.  These facts together
establish the theorem.

To demonstrate these facts, we simply observe that the invariant
($\spadesuit$) holds by induction on the time $t$ for the simulation we
defined.  This implies that at any point $t$, the number of pebbles on $D$
with respect to $P'$ is at most the number of nodes $v$ for which $b(v)
+ w(v) \geq 1/2$ with respect to $P$, and is therefore at most twice the
total value of pebbles with respect to $P$ at time $t$. Hence the cost
of pebbling $D$ using $P'$ is at most twice the cost of pebbling $D$
using $P$. Also, if there is a time $t$ at which the root $r$ has
black pebble value $1$ with respect to $P$, then $b(r) \ge 1/2$ at
time $t$, so there is a black pebble on $r$ with respect to $P'$ at
time $t$.
\end{proof}

The next result presents our best-known bounds for fractionally
pebbling trees $T^h_d$.

\begin{theorem}\label{t:daryFract}
\[(d-1)h/2 - d/2 \leq \FRpebbles(T^h_d) \leq (d-1)h/2 +1\]
\[\FRpebbles(T^3_d) = (3/2)d - 1/2\] 
\[\FRpebbles(T^4_2) = 3\]
\end{theorem}

We divide the proof into several parts.
First we prove the upper bound:
\[\FRpebbles(T^h_d) \leq (d-1)h/2 +1\]

\begin{proof} 
Let $A_h$ be the algorithm for
height $h \ge 2$. It is composed of two parts, $B_h$ and $C_h$. $B_h$
is run on the empty tree, and finishes with a black pebble on the
root and $(d-1)(h-2)$ white half pebbles below the root (and
of these $(d-1)(h-3)$ lie below the right child of the root).
Next, the black pebble
on the root is removed. Then $C_h$ is run on the result, and finishes
with the empty tree. $B_h$ and $C_h$ both use $(d-1)h/2 + 1$ pebbles.

$A_h'$ is the same as $A_h$ except that it finishes with a black half
pebble on the root. It does this in the most straight-forward way, by
leaving a black half pebble after the root is pebbled, and so it uses
$(d-1)h/2 + 1.5$ pebbles for all $h \ge 3$.

$B_2$: Pebble the tree of height 2 using $d$ black pebbles.

$B_h, h>2$: Run $A_{h-1}'$ on node 2 using $(d-1)(h-1)/2 + 1.5$ pebbles, and
then on node 3 (if $3 \leq d$) using a total of $(d-1)(h-1)/2 + 2$ pebbles (counting the
half pebble on node 2), and so on for nodes $2,3 \ldots, d$. So
$(d-1)(h-1)/2 + 1 + (d-1)/2 = (d-1)h/2 + 1$ pebbles are used when
$A_{h-1}'$ is run on node $d$.  Next run $B_{h-1}$ on node $d+1$, using
$(d-1)(h-1)/2 + 1$ pebbles on the subtree rooted at $d+1$, for
$(d-1)h/2 + 1$ pebbles in total
(counting the black half pebbles on node $2, \ldots, d$). The result
is a black pebble on node ${d+1}$, \ $(d-1)(h-3)$ white half pebbles under
$d+1$, and from earlier $d-1$ black half pebbles on
$2, \ldots, d$, for a total of $(d-1)(h-2)/2 + 1$ pebbles.
Add a white half pebble
to each of $2, \ldots, d$, then slide the black pebble from $d+1$
onto the root. Remove the black half pebbles from $2, \ldots, d$. Now
there are $(d-1)(h-2)$ white half pebbles under the root, and a black
pebble on the root.

$C_2$: The tree of height 2 is empty, so return.

$C_h$: The tree has no black pebbles and $(d-1)(h-2)$ white half
pebbles. Note that if a sequence can pebble a tree with $p$ pebbles,
then essentially the same sequence can be used to remove a white half
pebble from the root with $p + .5$ pebbles. $C_h$ runs $C_{h-1}$ on
node $d+1$, resulting in a tree with only a half white pebble on each
of $2, \ldots, d$. This takes $(d-1)h/2+1$ pebbles. Then $A_{h-1}$
is run on each of $2,\ldots,d$ in turn, to remove the white half
pebbles. The first such call of $A_{h-1}$ is the most expensive, using
$(d-1)(h-1)/2+1+(d-1)/2 = (d-1)h/2 + 1$ pebbles.
\end{proof}

\subsubsection*{}
As noted earlier, the tight lower bound for height 3 and any degree:
\[\FRpebbles(T^3_d) \ge 3/2d - 1/2\] 
follows from the asymptotically tight lower bound of $\Omega(k^{\frac{3}{2}d-\frac{1}{2}})$ states for 
nondeterministic branching programs computing $BT_d^3(k)$
(Corollary \ref{c:HtThree}).
We do, however, have a direct proof of $\FRpebbles(T^3_2) \ge 5/2$:

\begin{proof} Assume to the contrary that there is a fractional pebbling
with fewer than $2.5$ pebbles.  It follows that no non-leaf node $i$ can
ever have $w(i)\ge 0.5$, since the children of $i$ must each have pebble
value 1 in order to decrease $w(i)$.  Since there must be some time $t_1$
during the pebbling sequence such that both the nodes $2$ and $3$
(the two children of the root) have pebble value 1, it follows that at
time $t_1$, $b(2) > 0.5$ and $b(3)>0.5$.  Hence for $i=2,3$ there is a
largest $t_i \le t_1$ such that node $i$ is black-pebbled at time $t_i$ and
$b(i)>0.5$ during the time interval $[t_i,t_1]$.  (By `black-pebbled'
we mean at time $t_i-1$ both children of $i$ have pebble value 1,
so that at time $t_i$ the value of $b(i)$ can be increased.)

Assume w.l.o.g. that $t_2< t_3$.  Then at time $t_3-1$ both children of
node $3$ have pebble value 1 and $b(2)>0.5$, so the total pebble value
exceeds $2.5$. 
\end{proof}

\subsubsection*{}
Before we prove the lower bound for all heights, which we do not believe is tight,
we prove one more tight lower bound:
\[\FRpebbles(T^4_2) \ge 3\]
\begin{proof}
Let $C_0,C_1,\ldots,C_m$ be the sequence of pebble configurations
in a fractional pebbling of the binary tree of height 4.  We say
that $C_t$ is the configuration at time $t$.  Thus
$C_0$ and $C_m$ have no pebbles, and there is a first time $t_1$ such
that $C_{t_1+1}$ has a black pebble on the root.  In general we
say that step $t$ in the pebbling is the move form $C_t$ to $C_{t+1}$.
In particular, if an internal node $i$ is black-pebbled at step $t$
then both children of $i$ have pebble value 1 in $C_t$ and node
$i$ has a positive black pebble value in $C_{t+1}$.

Note that if any configuration $C_t$ has a whole white pebble
on some internal node then both children must have pebble value 1
to remove that pebble, so some configuration will have at least
pebble value 3, which is what we are to prove.  Hence we may
assume that no node in any $C_t$ has white pebble value 1,
and hence every node must be black-pebbled at some step.

For each node $i$ we associate a critical time $t_i$
such that $i$ is black-pebbled at step $t_i$ and hence the
children of $i$ each have pebble value 1 in configuration $C_{t_i}$.
The time $t_1$ associated with the root (as above) is the first
step at which the root is black-pebbled, and hence nodes 2 and 3
each have pebble value 1 in $C_{t_1}$.  In general if $t_i$
is the critical time for internal node $i$, and $j$ is a child
of $i$, then the critical time $t_j$ for $j$ is the largest $t<t_i$
such that $j$ is black-pebbled at step $t$. 

\medskip
\noindent
{\bf Sibling Assumption:}
We may assume w.l.o.g. (by applying an isomorphism to the tree)
that if $i$ and $j$ are siblings and $i<j$ then $t_i<t_j$.

In general the critical times for a path from root to leaf
form a descending chain.  In particular
$$
    t_7< t_3 < t_1
$$
For each $i>1$ we define $b_i$ and $w_i$ to be the black and
white pebble values of node $i$ at the critical time of its parent.
Thus for all $i>1$
\begin{equation}\label{e:bwsum}
    b_i + w_i =1
\end{equation}
Now let $p$ be the maximum pebble value of any configuration
$C_t$ in the pebbling.  Our task is to prove that $p\ge 3$

After the critical time of an internal node $i$ the white pebble
values of its two children must be removed.  When the first one
is removed both white values are present along with pebble value
1 on two children, so
$$
      w_{2i} + w_{2i+1} +2 \le p
$$
In particular for $i = 1,3$ we have
\begin{eqnarray}
    w_2 + w_3 +2 & \le & p \label{e:wtwothree} \\
    w_6 + w_7 +2 & \le & p \label{e:wsevtwo}
\end{eqnarray}

Now we consider two cases, depending on the order of $t_2$ and $t_7$.

\medskip
\noindent
{\bf CASE I:}  $t_2<t_7$

Then by the Sibling Assumption, at time $t_7$ (when node 7
is black-pebbled) we have
\begin{equation}\label{e:twosix}
   b_2+b_6 +2 \le p
\end{equation}
Now if we also suppose that $w_6$ is not removed until after
$t_1$ (CASE IA) then when the first of $w_2,w_6$ is removed we have
$$
   w_2+ w_6 +2 \le p
$$ 
so adding this equation with (\ref{e:twosix}) and using
(\ref{e:bwsum}) we see that $p\ge 3$ as required.

However if we suppose that $w_6$ is removed before $t_1$
(CASE IB) (but necessarily after $t_2 < t_3$) then we have
$$
  b_2+b_3 + w_6 +2 \le p
$$
then we can add this to (\ref{e:wtwothree}) to again
obtain $p \ge 3$.

\medskip
\noindent
{\bf CASE II:} $t_7< t_2$    

Then $t_6<t_7<t_2<t_3$ so at time $t_2$ we have
$$
    b_6+b_7 +2 \le p
$$
so adding this to (\ref{e:wsevtwo}) we again obtain $p\ge 3$.
\end{proof}

\subsubsection*{}
To prove the general lower bound, we need the following lemma:
\begin{lem}\label{l:rational} For every finite DAG there is an optimal
fractional B/W pebbling in which all pebble values are rational numbers.
(This result is robust independent of various definitions of pebbling;
for example with or without sliding moves, and whether or not we require
the root to end up pebbled.) 
\end{lem}
\begin{proof}  Consider an optimal B/W fractional pebbling algorithm.
Let the variables $b_{v,t}$ and $w_{v,t}$ stand for the black and white
pebble values of node $v$ at step $t$ of the algorithm.

{\bf Claim:}  We can define a set of linear inequalities with 0 -
1 coefficients  which suffice to ensure that the pebbling is legal.

For example,  all variables are non-negative, $b_{v,t} + w_{b,t} \le 1$,
initially all variables are 0, and finally the nodes have the values
that we want, node values remain the same on steps in which nothing is
added or subtracted, and if the black value of a node is increased at
a step then all its children must be 1 in the previous step, etc.

Now let $p$ be a new variable representing the maximum pebble value of
the algorithm.  We add an inequality for each step $t$ that says the
sum of all pebble values at step $t$ is at most $p$.

Any solution to the linear programming problem:

    Minimize $p$ subject to all of the above inequalities

gives an optimal pebbling algorithm for the graph.  But every LP program
with rational coefficients has a rational optimal solution (if it has
any optimal solution). 
\end{proof}

\subsubsection*{}
Now we can prove the lower bound for all heights:
\[\FRpebbles(T^h_d) \geq (d-1)h/2 - d/2 \]
\begin{proof} 

The high-level strategy for the proof is as follows. Given $d$ and $h$, we 
transform the tree $T_{d}^{h}$ into a DAG $G_{d,h}$ such that a lower
bound on $\BWpebbles(G_{d,h})$ gives a lower
bound for $\FRpebbles(T_{d}^{h})$. 
To analyze $\BWpebbles(G_{d,h})$, we use a 
result of
Klawe \cite{klawe}, who shows that for a DAG $G$ that satisfies a certain 
``niceness'' property, $\BWpebbles(G)$ can be given in terms of $\Bpebbles(G)$ 
(and the relationship is tight to within a constant less than one). 
The black pebbling cost is typically
easier to analyze. In our case, $G_{d,h}$ does not satisfy the niceness
property as-is, but just by removing some edges from $G_{d,h}$, we get
a new DAG $G'_{d,h}$ which is nice. We then show how to exactly compute
$\Bpebbles(G'_{d,h})$ which yields a lower bound on
$\BWpebbles(G_{d,h})$, and hence on $\FRpebbles(T_{d}^{h})$.

We first motivate the construction $G_{d,h}$ and show that the whole black-white
pebbling number of $G_{d,h}$ is related to the fractional pebbling number
of $T_{d}^{h}$. 
 
We first use Lemma \ref{l:rational} to ``discretize'' the
fractional pebble game. The following are the rules for the
discretized game, where $c$ is a parameter:
\begin{packed_item}
 \item For any node $v$, decrease $b(v)$ or increase $w(v)$ by $1/c$.
 \item For any node $v$, including leaf nodes, if all the children of $v$
 have value 1, then increase $b(v)$ or decrease $w(v)$ by $1/c$.
\end{packed_item}

By Lemma \ref{l:rational}, we can assume all pebble
values are rational, and if we choose $c$ large enough it is not a restriction
that pebble values can only be changed by $1/c$. Since sliding moves
are not allowed, the pebbling cost for this game is at most one more
than the cost of fractional pebbling with black sliding moves.

Now we show how to construct $G_{d,h}$ (for an example, see figure \ref{f:reductionG}).
 We will split up each node of
$T_{d}^{h}$ into $c$ nodes, so that the discretized game corresponds to
the whole black-white pebble game on the new graph. Specifically, the cost
of the whole black-white pebble game on the new graph will be exactly $c$
times the cost of the discretized game on $T_{d}^{h}$.

In place of each node $v$ of $T_{d}^{h}$, $G_{d,h}$ has $c$ nodes $v[1], \ldots, v[c]$;
having $c'$ of the $v[i]$ pebbled simulates $v$ having value $c'/c$. In
place of each edge $(u,v)$ of $T_{d}^{h}$ is a copy of the complete bipartite 
graph $(U,V)$, where $U$ contains nodes $u[1] \ldots u[c]$ and $V$ contains nodes $v[1]
 \ldots v[c]$. If $u$ was a parent of $v$ in the tree, then all the edges go
from $V$ to $U$ in the corresponding complete bipartite graph. Finally, a new 
``root'' is added at height $h+1$ with edges from each
of the $c$ nodes at height $h$\footnote{The reason for this is quite technical: Klawe's definition 
of pebbling is slightly different from ours in that it requires that the root remain pebbled. Adding 
a new root forces there to be a time when all $c$ of the height $h$ nodes, which represent the
root of $T_d^h$, are pebbled. Adding one more pebble to $G_{d,h}$ changes the relationship between 
the cost of pebbling $T_d^h$ and the cost of pebbling $G_{d,h}$ by a negligible amount.}. 
So every node at height $h-1$ and lower
has $c$ parents, and every internal node except for the root has $dc$
children. 

\ifpdf
\else
\begin{figure}
\vspace*{.3cm}
\hspace*{1.5cm}\includegraphics[scale=0.70]{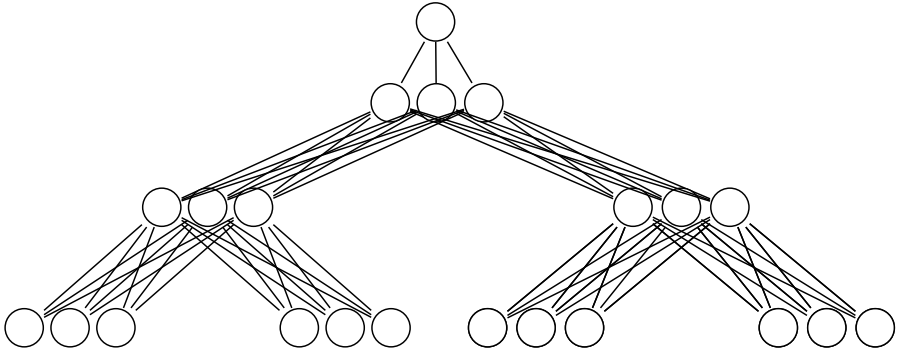}
\hspace*{1cm} \vspace*{.2cm}
\caption{$G_{2,3}$ with $c=3$}
\label{f:reductionG}
\end{figure}
\fi

To lower bound $\BWpebbles(G_{d,h})$, we will use
Klawe's result \cite{klawe}. Klawe showed that for ``nice'' graphs $G$, the black-white
pebbling cost of $G$ (with black and white sliding moves) is at least $\lfloor \Bpebbles/2 
\rfloor + 1$. Of course, the black-white
pebbling cost without sliding moves is at least the cost with them. We define 
what it means
for a graph to be nice in Klawe's sense.

\begin{defn}
\label{d:nice}
A DAG $G$ is nice if the following conditions hold:
\begin{enumerate}
\item If $u_1$, $u_2$ and $u$ are nodes of $G$ such that $u_1$ and $u_2$
are children of $u$ (i.e., there are edges from $u_1$ and $u_2$ to $u$),
then the cost of black pebbling $u_1$ is equal to the cost of black
pebbling $u_2$
\item If  $u_1$ and $u_2$ are children of $u$, then there is no path from  $u_1$ to $u_2$ or from $u_2$ to $u_1$.  
\item If $u, u_1, \ldots, u_m$ are nodes none of which has
 a path to another, then there are node-disjoint paths $P_1, \ldots,
 P_m$ such that $P_i$ is a path from a leaf (a node with in-degree 0) to $u_i$ 
and there is no path between $u$ and any node in $P_i$.
\end{enumerate}
\end{defn}

$G_{d,h}$ is not nice in Klawe's sense. We will delete some edges from
$G_{d,h}$ to produce a nice graph $G'_{d,h}$ and we will analyze $\Bpebbles(G'_{d,h})$. 
Note that a lower bound on $\BWpebbles(G'_{d,h})$ is also a lower bound on $\BWpebbles(G_{d,h})$.

The following definition will help in explaining the construction of
$G'_{d,h}$ as well as for specifying and proving properties of certain paths. 

\begin{defn} For $u \in G_{d,h}$, let $T_{d}^{h}(u)$ be the node in $T_{d}^{h}$ such that
$T_{d}^{h}(u)[i] = u$ for some $i \leq c$.  For $v,v' \in T_{d}^{h}$, we say $v < v'$
if $v$ is visited before $v'$ in an inorder traversal of $T_{d}^{h}$. For $u,u'
\in G_{d,h}$, we say $u < u'$ if $T_{d}^{h}(u) < T_{d}^{h}(u')$ or if for some $v \in T_{d}^{h}$, $u =
v[i]$, $u' = v[j]$, and $i < j$. 
\end{defn}

$G_{d,h}'$ is obtained from $G_{d,h}$ by removing $c-1$ edges from each internal node
except the root, as follows (for an example, see figure \ref{f:reductionGprime}). 
For each internal node $v$ of $T$, consider
the corresponding nodes $v[1], v[2], \ldots, v[c]$ of $G_{d,h}$. Remove the
edges from $v[i]$ to its $i-1$ smallest and $c-i$ largest children. So
in the end each internal node except the root has $c(d-1)+1$ children.


\ifpdf
\else
\begin{figure}
\vspace*{.3cm}
\hspace*{1.5cm}\includegraphics[scale=0.70]{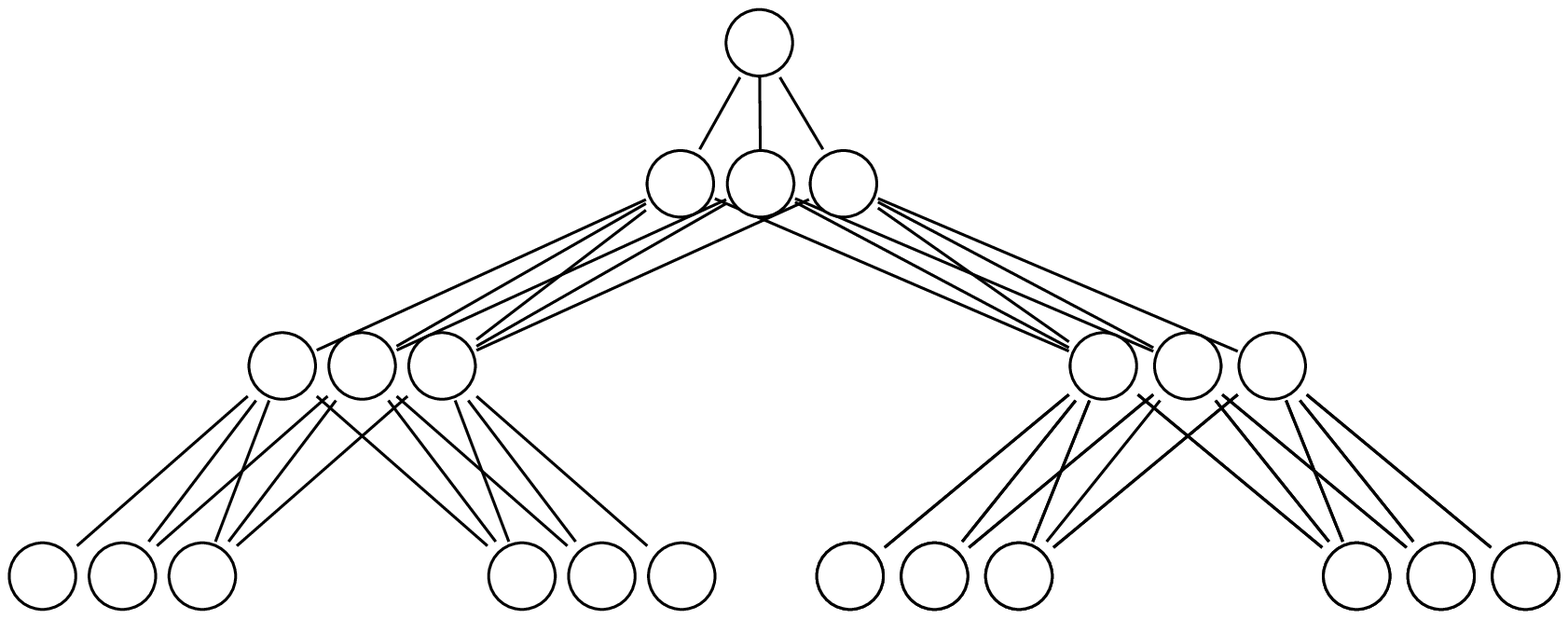}
\hspace*{1cm} \vspace*{.2cm}
\caption{$G_{2,3}'$ with $c=3$}
\label{f:reductionGprime}
\end{figure}
\fi

We first analyze \Bpebbles($G'_{d,h})$ and then show that
it is nice.
We show that $\Bpebbles(G_{d,h}') = c[(d-1)(h-1) +
1]$. Note that an upper bound of $c[(d-1)(h-1) +
1]$ is attained using a simple recursive algorithm similar to that used
for the binary tree.

For the lower bound, consider the earliest time $t$ when all paths from
 a leaf to the root are blocked. Figure \ref{f:fract_bottleneck} 
 is an example of the type of
 pebbling configuration that we are about to analyze. 
 The last pebble placed must have been
placed at a leaf, since otherwise $t-1$ would be an earlier time when
all paths from a leaf to the root are blocked. Let $P$ be the
newly-blocked path from a leaf to the root.  Consider the set $S = \{
u \in G_{d,h}' \ \vert\ u \not \in P \text{ and $u$ is a child of a node in }
P \}$ of size $c (d-1)(h-1) + (c-1) = c[(d-1)(h-1) + 1] - 1$ (the $c-1$
is contributed by nodes at height $h$).  We will give a set of
mutually node-disjoint paths $\{P_u\}_{u \in S}$ such that $P_u$ is a
path from a leaf to $u$ and $P_u$ does not intersect $P$. At time $t-1$,
there must be at least one pebble on each $P_u$, since otherwise there would
still be an open path from a leaf to the root at time $t$. Also counting
the leaf node that is pebbled at $t$ gives c[(d-1)(h-1) + 1] pebbles.

\begin{defn}
 The left-most (right-most) path to $u$ is the unique path ending at
 $u$ determined by choosing the smallest (largest) child at every level.
\end{defn}
\begin{defn}
 $P(l)$ is the node of path $P$ at height $l$, if it exists.
\end{defn}

For each $u \in S$ at height $l$, if $u$ is less than (greater than)
$P(l)$ then make $P_u$ the left-most (right-most) path to $u$. Now
we need to show that the paths $\{P_u\}_{u \in S} \cup \{P\}$ are
disjoint. The following fact is clear from the definition of $G_{d,h}'$.
\begin{lem}\label{l:thefact} For any $u,u' \in G_{d,h}'$, if $u < u'$ then the
smallest child of $u$ is not a child of $u'$, and the largest child of
$u'$ is not a child of $u$. 
\end{lem}

First we show that $P_u$ and $P$ are disjoint. The following lemma will
help now and in the proof that $G'_{d,h}$ is nice.

\begin{lem}\label{l:paths}
 For $u,v \in G_{d,h}'$ with $u < v$, if there is no path from $u$ to $v$
 or from $v$ to $u$ then the left-most path to $u$ does not intersect
 any path to $v$ from a leaf, and the right-most path to $v$ does not
 intersect any path to $u$ from a leaf.
\end{lem}
\begin{proof} Suppose otherwise and let $P_u'$ be the left-most
path to $u$, and $P_v'$ a path to $v$ that intersects $P_u'$. Since
there is no path between $u$ and $v$, there is a height $l$, one
greater than the height where the two paths first intersect, such that
$P_u'(l), P_v'(l)$ are defined and $P_u'(l) < P_v'(l)$. But then from
Lemma \ref{l:thefact} $P_u'(l-1) \not = P_v'(l-1)$, a contradiction. The
proof for the second part of the lemma is similar. 
\end{proof}

That $P_u$ and $P$ are disjoint follows from using Lemma \ref{l:paths}
on $u$ and the sibling of $u$ in $P$. 

Next we show that for distinct
$u,u' \in S$, $P_u$ does not contain $u'$. Suppose it does. Assume $P_u$
is the left-most path to $u$ (the other case is similar). Since $u
\not = u'$, there must be a height $l$ such that $P_u(l)$ is defined
and $P_u(l-1) = u'$. From the definition of $S$, we know $P(l)$ is also
a parent of $u'$. From the construction of $P_u$, since we assumed $P_u$
is the left-most path to $u$, it must be that $P_u(l) < P(l)$. But then
Lemma \ref{l:thefact} tells us that $u'$ cannot be a child of $P(l)$, a
contradiction. 

The proof that $P_u$ and $P_{u'}$ do not intersect is by contradiction.
Assuming that there are $u,u' \in S$ such that
$P_u$ and $P_{u'}$ intersect, there is a height $l$, one greater
than the height where they first intersect, such that $P_u(l) \not =
P_{u'}(l)$. Note that $P_u$ and $P_{u'}$ are both left-most paths or both
right-most paths, since otherwise in order for them to intersect they
would need to cross $P$. But then from Lemma \ref{l:thefact} $P_u(l-1)
\not = P_{u'}(l-1)$, a contradiction.

This is an example of a bottleneck of the specified structure for $G_{d,h}'$
corresponding to the height 3 binary tree, with $c=3$:

\ifpdf
\else
\begin{figure}
\vspace*{.3cm}
\hspace*{1.5cm}\includegraphics[scale=0.70]{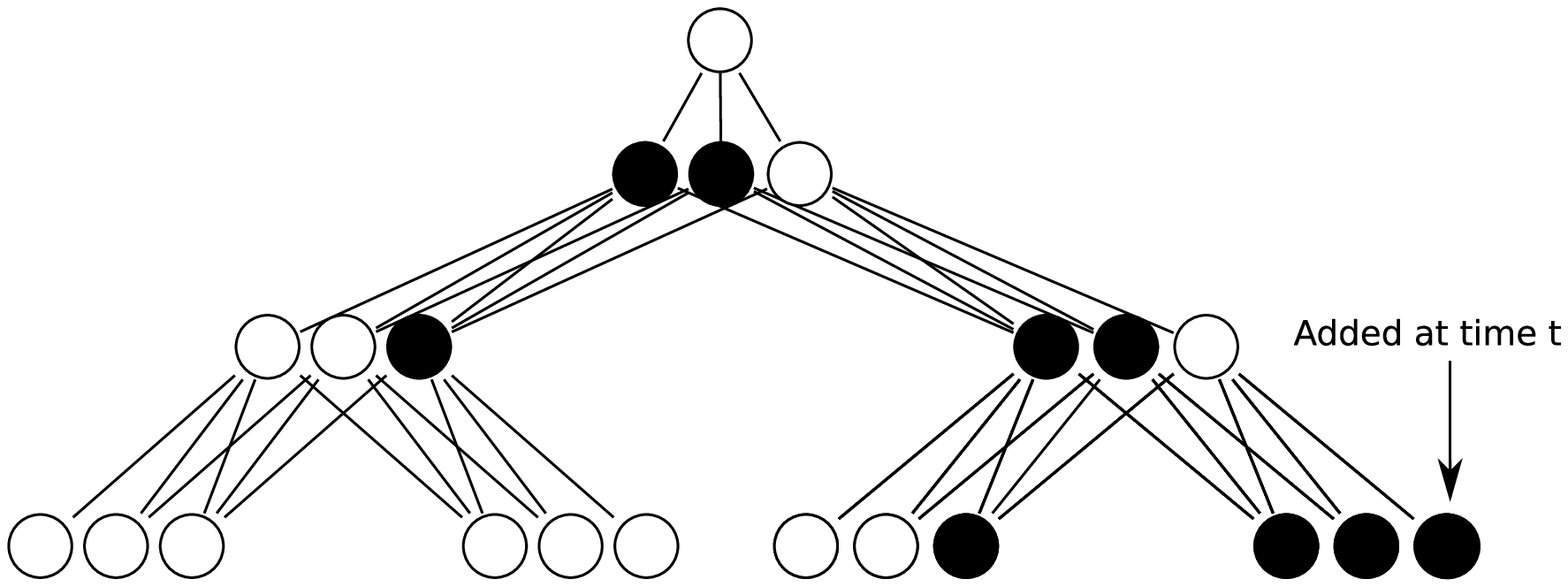}
\hspace*{1cm} \vspace*{.4cm}
\caption{A possible black pebbling bottleneck of $G_{2,3}'$, with $c=3$}
\label{f:fract_bottleneck}
\end{figure}
\fi

The last step is to prove that $G_{d,h}'$ is nice. There are three
properties specified in Definition \ref{d:nice}.
Property 2 is obviously satisfied. For property 1, the argument used to
give the black pebbling lower bound of $c[(d-1)(h-1) + 1]$ can be used
to give a black pebbling lower bound of $c(d-1)(l-1) + 1$ for any node
at height $l \leq h$ (the 1 is for the last node pebbled, and recall
the root is at height $h+1$), and that bound is tight. For property 3,
choose $P_i$ to be the left-most (right-most) path from $u_i$ if $u_i$
is less than (greater than) $u$. Then use Lemma \ref{l:paths} on each
pair of nodes in $\{u, u_1,\ldots, u_m\}$. 

Since $\Bpebbles(G'_{d,h}) = c[(d-1)(h-1)+ 1]$, we have
\[\BWpebbles(G_{d,h}) \geq \BWpebbles(G'_{d,h}) \geq c[(d-1)(h-1) + 1]/2\] and 
thus that the pebbling cost for the discretized game on $T_{d}^{h}$ is
at least $(d-1)(h-1)/2 + .5$, which implies $\FRpebbles(T_{d}^{h}) \geq
(d-1)(h-1)/2 - .5$.
\end{proof}

\subsection{White sliding moves}\label{s:wSlide}

In the definition of fractional pebbling (Definition \ref{d:pebbling})
we allow black sliding moves but not white sliding moves.
To allow white sliding moves we would add a clause

\medskip
\noindent
(iv)  For every internal node $i$, if $w(i)=1$ and $j$ is a
child of $i$ and every child of $i$ except $j$ has total pebble
value 1, then decrease $w(i)$ to 0 and increase $w(j)$ so that
node $j$ has total pebble value 1.

We did not include this move in the original definition because
a nondeterministic $k$-way BP solving \ft\ or \bt\
does not naturally simulate it.
The natural way to simulate such a move would be to verify
the conjectured value of node $i$ (conjectured when the white
pebble was placed on $i$) by comparing it with
$f_i(v_{j_1},\ldots,v_{j_d})$,
where $j_1,\ldots,j_d$ are the children of $i$.  But this
would require the BP to remember a $(d+1)$-tuple of values,
whereas potentially only $d$ pebbles are involved.

White sliding moves definitely reduce the number of pebbles required
to pebble some trees.
For example the binary tree $T^3_2$ can easily be pebbled with 2 pebbles
using white sliding moves, but requires 2.5 pebbles without (Theorem
\ref{t:daryFract}).  The next result shows that $8/3$ pebbles
suffice for pebbling $T^4_2$ with white sliding moves, whereas
3 pebbles are required without (Theorem \ref{t:daryFract}).

\begin{theorem}\label{t:eight-thirds} The binary tree of height 4 can
be pebbled with $8/3$ pebbles using white sliding moves.
\end{theorem}
\begin{proof} The height 3 binary tree can be pebbled with 2 pebbles. Use
that sequence on node 2, but leave a third black pebble on node 2. That
takes 7/3 pebbles. Put black pebbles on nodes 12 and 13. Slide a third
black pebble up to node 6. Remove the pebbles on nodes 12 and 13. Put
black pebbles on nodes 14 and 15 -- this is the first configuration with
8/3 pebbles. Slide the pebble on node 14 up to node 7. Remove the pebble
from 15. Put 2/3 of a white pebble on node 6. Slide the black pebble
on node 7 up to node 3. Remove a third black pebble from node 6. Put
2/3 of a white pebble on node 2 -- the resulting configuration has 8/3
pebbles. Slide the black pebble on node 3 up to the root. Remove all
black pebbles. At this point there is 2/3 of a white pebble on both node
2 and node 6. Put a black pebble on node 12 and a third black pebble
on node 13 -- another bottleneck. Slide the 2/3 white pebble on node
6 down to node 13. Remove the pebbles from nodes 12 and 13. Finally,
use 8/3 pebbles to remove the 2/3 white pebble from node 2.
\end{proof}

\section{Branching Program Bounds}\label{s:PBbounds}

In this section we prove tight bounds (up to
a constant factor) for the number of states required for
both deterministic and nondeterministic $k$-way  branching
programs to solve the Boolean problems \btdthree\ for all
trees of height 2 and 3. (The bound is obviously $\Theta(k^d)$ for
trees of height 2, because there are $d + k^d$ input variables.)
For every height $h\ge 2$
we prove upper bounds for deterministic {\em thrifty} programs
which solve \ft\
(Theorem \ref{t:BPUpper}, (\ref{e:dFUpper})), and show that these
bounds are optimal for degree $d=2$ even for the Boolean problem \bt\
(Theorem \ref{t:detThriftLB}).  We prove upper bounds for
nondeterministic thrifty programs solving \bt\ in general,
and show that these are optimal for binary trees of height 4
or less (Theorems \ref{t:BPUpper} and \ref{t:thrifFourtwo}). 

For the nondeterministic case our best BP upper bounds for every
$h\ge 2$ come from
fractional pebbling algorithms via Theorem \ref{t:pebSim}.
For the deterministic case our best bounds for the function
problem \ft\ come from black pebbling via the same theorem,
although we can improve on them for the Boolean problem
\bttwoh\ by a factor of $\log k$ (for $h\ge 3)$.

\begin{theorem}[BP Upper Bounds]
\label{t:BPUpper}
For all $h,d\ge 2$
\begin{eqnarray}
   \dFstate & = & O(k^{(d-1)h-d+2})   \label{e:dFUpper} \\
   \dBstate & = &  O(k^{(d-1)h-d+2} / \log k), \mbox{ for $h\ge 3$}
                                       \label{e:dBUpper}   \\
   \ndBstate &  = & O(k^{(d-1)(h/2)+1}) \label{e:nBUpper} 
\end{eqnarray}
The first and third bounds are realized by thrifty programs.
\end{theorem}

\begin{proof}
The first and third bounds follow from Theorem \ref{t:pebSim}
(which states that pebbling upper bounds give rise to upper
bounds for the size of thrifty BPs) and from Theorems
\ref{t:blackSliding} and \ref{t:daryFract} (which give the
required pebbling upper bounds).

To prove (\ref{e:dBUpper}) we use a branching program which
implements the algorithm below.   Here we have a
parameter $m$, and choosing
$m= \lceil \log k^{d-1} - \log\log k^{d-1} \rceil$ suffices
to show $\dBstate = O(k^{(d-1)(h-1)+1} / \log k^{d-1})$,
from which (\ref{e:dBUpper}) follows.
We estimate the number of states required up to a constant factor.  

\medskip
\noindent
1)  Compute $v_2$ (the value of node 2 in the heap ordering),
using the black pebbling algorithm for the principal left subtree.
This requires $k^{(d-1)(h-2)+1}$ states.
Divide the $k$ possible values
for $v_2$ into $\lceil k/m \rceil$ blocks of size $m$.

\medskip
\noindent
2)  Remember the block number for $v_2$, and compute
$v_3,\ldots,v_{d+1}$.  This requires
$k/m \times k^{d-2} \times k^{(d-1)(h-2) + 1} = k^{(d-1)(h-1)+1}/m$
states.

\medskip
\noindent
3)  Remember $v_3,\ldots,v_{d+1}$ and the block number for $v_2$.
Compute $f_1(a, v_3, \ldots, v_{d+1})$ for each of the $m$ possible values $a$ for $v_2$
in its block number, and keep track of the set of $a$'s for which
$f_1 = 1$.  This requires
$k^{d-1} \times k/m \times m \times 2^m = k^d  2^m$ states.

\medskip
\noindent
4)  Remember just the set of possible $a$'s (within
its block) from above (there are $2^m$ possibilities).  Compute $v_2$
again and accept or reject depending on whether $v_2$ is in the subset.
This requires $k^{(d-1)(h-2)+1}  2^m$ states.

The total number of states has order the maximum of
$ k^{(d-1)(h-1)+1}/m$ and $k^{(d-1)(h-2)+1}  2^m$,
which is at most
$$
k^{(d-1)(h-1)+1} / (\log k^{d-1} - \log \log k^{d-1})
$$
for $m = \log k^{d-1} - \log \log k^{d-1}$.
\end{proof}

We combine the above upper bounds with the Ne\u{c}iporuk lower
bounds in Subsection \ref{s:NecLB}, Figure \ref{neci},
to obtain the following.

\begin{cor}[Tight bounds for height 3 trees]
\label{c:HtThree} For all $d\ge 2$
\begin{eqnarray*}
    \dFdthreestate & = &  \Theta(k^{2d-1}) \\
    \dBdthreestate & = & \Theta(k^{2d-1}/\log k)  \\
    \nddthreeBstate & = & \Theta(k^{(3/2)d - 1/2})  
\end{eqnarray*}
\end{cor}

\subsection{The Ne\u{c}iporuk method}\label{s:NecLB}

The Ne\u{c}iporuk method still yields the strongest explicit
binary branching program size lower bounds known today, namely
$\Omega(\frac{n^2}{(\log n)^2})$ for deterministic \cite{ne66} and
$\Omega(\frac{n^{3/2}}{\log n})$ for nondeterministic (albeit for a
weaker nondeterministic model in which states have bounded outdegree
\cite{pu87}, see \cite{ra91}).

By \emph{applying the Ne\u{c}iporuk method} to a $k$-way branching program
$B$ computing a function $f:[k]^m \rightarrow R$, we mean the
following well known steps \cite{ne66}:
\begin{enumerate}
\item Upper bound the
  number $N(s,v)$ of (syntactically) distinct branching programs of
  type $B$ having $s$ non-final states, each labelled by one of $v$ variables.
\item Pick a partition $\{V_1,\ldots, V_p\}$ of $[m]$.
\item For $1\leq i\leq p$, lower bound the number $r_{V_i}(f)$
  of restrictions $f_{V_i}: [k]^{|V_i|} \rightarrow R$
  of $f$ obtainable by fixing values of the variables in
  $[m]\setminus V_i$.
\item Then size($B$) $\geq$ $|R|+\sum_{1\leq i\leq p} s_i$, where $s_i
  = \min \{\ s : N(s,|V_i|) \geq r_{V_i}(f)\ \}$.
\end{enumerate}

\begin{theorem} \label{necilowerbound} Applying the Ne\u{c}iporuk
  method yields Figure \ref{neci}.
\end{theorem}

\renewcommand*\arraystretch{1.5} 

\begin{figure}
\hspace*{-4mm}
\begin{tabular}{|p{.25\textwidth}||p{.3\textwidth}|p{.3\textwidth}|}
\hline \multicolumn{1}{|c||}
{Model}
 & \multicolumn{1}{|c||}{Lower bound for \ft}
  &
\multicolumn{1}{|c||}{Lower bound for \bt}
\\ \hline \hline
Deterministic $k$-way branching program &
\multicolumn{1}{c|}{\raisebox{-1ex}{\framebox{$\frac{d^{h-2}-1}{4(d-1)^2}\cdot
    k^{2d-1}$}}} &
\multicolumn{1}{c|}{\raisebox{-1ex}{\framebox{$\frac{d^{h-2}-1}{3(d-1)^2}\cdot
    \frac{k^{2d-1}}{\log k}$}}} \\
\hline
Deterministic binary branching program &
\multicolumn{1}{c|}{\raisebox{-1ex}{
$\frac{d^{h-2}-1}{5(d-1)^2} \cdot
    k^{2d} = \Omega(n^2/(\log n)^2)$
}} &
\multicolumn{1}{c|}{\raisebox{-1ex}{
$\frac{d^{h-2}-1}{4d(d-1)} \cdot
    \frac{k^{2d}}{\log k} = \Omega(n^2/(\log n)^3)$
}} \\
\hline
Nondeterministic $k$-way BP &
\multicolumn{1}{c|}{\raisebox{-1ex}{
 $\frac{d^{h-2}-1}{2d-2} \cdot k^{\frac{3d}{2}-\frac{1}{2}}\sqrt{\log k}$
}}
&
\multicolumn{1}{c|}{\raisebox{-1ex}{\framebox{$\frac{d^{h-2}-1}{2d-2}\cdot
    k^{\frac{3d}{2}-\frac{1}{2}}$}}} \\
\hline
Nondeterministic binary BP &
\multicolumn{1}{c|}{\raisebox{-1ex}{$\frac{d^{h-2}-1}{2d-2} \cdot
    k^{\frac{3d}{2}}\sqrt{\log k} = \Omega(n^{3/2}/\log n)$}} &
\multicolumn{1}{c|}{\raisebox{-1ex}{$\frac{d^{h-2}-1}{2d-2} \cdot
    k^{\frac{3d}{2}}= \Omega(n^{3/2}/(\log n)^{3/2})$}} \\
\hline
\end{tabular}
\caption{Size bounds,
  expressed in terms of
  $n=\Theta(k^d\log k)$ in the binary cases, obtained by applying the
  Ne\u{c}iporuk method. Rectangles indicate
  optimality in $k$ when $h=3$ (Cor.\ \ref{c:HtThree}).
Improving any entry to
  $\Omega(k^{\mbox{\scriptsize unbounded }f(h)})$ would
  prove $\lspace\subsetneq \p$ (Cor.\ \ref{c:thegoal}).}
  \label{neci}
\end{figure}
\begin{remark}
  The $\Omega(n^{3/2}/(\log n)^{3/2})$ binary nondeterministic BP
  lower bound for the \bt\ problem and in particular for \bttwothree\
  applies to the number of \emph{states} when these can have arbitrary
  outdegree.  This seems to improve on the best known former bound of
  $\Omega(n^{3/2}/\log n)$, slightly larger but obtained for the weaker
  model in which states have bounded degree, or equivalently, for the
  switching and rectifier network model in which size is defined as
  the number of edges \cite{pu87,ra91}.
\end{remark}

\begin{proof}[Proof of Theorem \ref{necilowerbound}]
 We have $\dkway(s,v)\leq v^s \cdot (s+|R|)^{sk}$ for the
  number of deterministic BPs and
  $\nkway(s,v) \leq v^s \cdot (|R|+1)^{sk}\cdot(2^s)^{sk}$ for
  nondeterministic BPs
  having $s$ non-final states, each labelled with one of $v$
  variables. To see $\nkway(s,v)$, note that
  edges labelled $i\in[k]$ can connect a state $S$ to zero or one state
  among the final states and can connect $S$ independently to any
  number of states among the non-final states.

  The only decision to make when applying the Ne\u{c}iporuk method is
  the choice of the partition of the input variables. Here every entry
  in Figure \ref{neci} is obtained using the same partition
  (with the proviso that a $k$-ary variable in the partition is
  replaced by $\log k$ binary variables when we treat $2$-way branching
  programs).

  We will only partition the set $V$ of $k$-ary  \ft\ or
  \bt\ variables that pertain to
  internal tree nodes other than the root (we will neglect the root
  and leaf variables). Each
  internal tree node has $d-1$ siblings and each
  sibling involves $k^d$ variables.  By a \emph{litter} we will mean
  any set of $d$ $k$-ary
  variables that pertain to precisely $d$ such siblings. We obtain our
  partition by writing $V$ as a union of
$$k^d \cdot
  \Sigma_{i=0}^{h-3}d^i = k^d \cdot \frac{d^{h-2}-1}{d-1} $$ litters.
  (Specifically, each litter can be
  defined as
$$\{f_i(j_1,j_2,\ldots,j_d),f_{i+1}(j_1,j_2,\ldots,j_d),\ldots,f_{i+d-1}(j_1,j_2,\ldots,j_d)\}$$
for some $1\leq j_1,j_2,\ldots,j_d\leq k$ and some $d$ siblings
$i,i+1,\ldots,i+d-1$.)

Consider such a litter $L$.
We claim that $|R|^{k^d}$ distinct functions $f_L :
[k]^d \rightarrow R$ can be induced by setting the variables outside of $L$,
where $|R|=k$ in the case of \ft\ and $|R|=2$ in the case of \bt.
Indeed, to induce any such function, fix the
``descendants of the litter $L$'' to make each variable in $L$ relevant
to the output; then, set the variables pertaining to the immediate
ancestor node $\nu$ of the siblings forming $L$
to the appropriate $k^d$ values, as if those were the final output
desired; finally, set
all the remaining variables in a way such that
the values in $\nu$  percolate from $\nu$ to the root.

It remains to do the calculations. We illustrate two cases.
Similar calculations yield the other entries in Figure \ref{neci}.

\noindent
\emph{Nondeterministic $k$-way branching programs computing \ft}.
Here $|R|=k$. In a correct program, the number $s$ of states querying
one of the $d$ litter $L$ variables must satisfy
$$
k^{k^d} \leq \nkway(s,d) \leq d^s \cdot (k+1)^{sk}\cdot(2^s)^{sk} \leq
s^s \cdot k^{2sk}\cdot(2^s)^{sk}
$$
since $d\leq s$ (because \ft\
depends on all its variables), and thus
$$
k^d\log k \leq s(\log s + 2k\log k) + s^2k.
$$
Suppose to the contrary that
$s< (k^{\frac{d-1}{2}}\sqrt{\log k})/2$. Then
$$
s(\log s + 2k\log k) + s^2k
< s (\frac{d-1}{2}\log k + \frac{\log\log k}{2} + 2k\log k) + s^2k
< s(sk) + s^2k
< k^d\log k
$$
for large $k$ and all $d\geq 2$, a contradiction.
Hence $s\geq
(k^{\frac{d-1}{2}}\sqrt{\log k})/2$.
Since this holds for every litter, recalling step 4 in the
Ne\u{c}iporuk method as described prior to Theorem~\ref{necilowerbound},
the total number of states in the program is at least
$$
k + k^d \cdot \frac{d^{h-2}-1}{d-1} \cdot
(k^{\frac{d-1}{2}}\sqrt{\log k})/2
\geq
\frac{d^{h-2}-1}{2d-2} \cdot k^{\frac{3d}{2}-\frac{1}{2}}\sqrt{\log k}.
$$

\noindent
\emph{Nondeterministic binary (ie $2$-way)
branching programs deciding \bt}.
Here $|R|=2$. When the program is binary, the $d$ variables in the
litter $L$ become $d \log k$ Boolean variables. The number $s$ of
states querying one of these $d \log k$ variables then verifies
$$
2^{k^d} \leq \ntwoway(s,d\log k) \leq (d\log k)^s\cdot
        (2+1)^{2s} \cdot (2^s)^{2s} < (s\log k)^s \cdot 2^{4s+2s^2}
$$
since $d\leq s$ and thus
$$
k^d \leq s\log s + s\log\log k + 4s + 2s^2 \leq 3s^2 + 5s \log\log k.
$$
It follows that $s\geq k^{\frac{d}{2}}/2$.
Hence the total number of states in a binary
nondeterministic program deciding \bt\ is at least
$$
k^d \cdot \frac{d^{h-2}-1}{d-1} \cdot  \frac{k^{d/2}}{2} \geq
\frac{d^{h-2}-1}{2(d-1)} \cdot k^{\frac{3d}{2}} =
\frac{d^{h-2}-1}{2(d-1)} \cdot \frac{(k^d\log k)^{3/2}}{(\log k)^{3/2}} =
\Omega(n^{3/2}/(\log n)^{3/2})
$$
where $n=\Theta(k^d\log k)$ is the length of the binary
encoding of \bt.
\end{proof}

\medskip

The next two results show limitations on the Ne\u{c}iporuk method
that are not necessarily present in the state sequence method
(see Theorems \ref{t:childLB} and \ref{t:beatittwice}).

Let \func\ have the same input as \ft\ with the exception that the
root function is deleted.
The output is the tuple $(v_2, v_3,\ldots, v_{d+1})$ of values for the
children of the root.  \func\ can be computed by a $k$-way
deterministic BP with $O(k^{(d-1)h-d+2})$ states using the same
black pebbling method which yields the bound (\ref{e:dFUpper})
in Theorem \ref{t:BPUpper}.

\begin{theorem} \label{t:rootfunction}
For any $d,h\geq 2$, the best $k$-way deterministic BP
size lower bound attainable for \func\ by applying the
Ne\u{c}iporuk method is $\Omega(k^{2d-1})$.
\end{theorem}

\begin{proof}
  The function $\func: [k]^m \rightarrow R$ has
  $m=\Theta(k^d)$. Any partition $\{V_1,\ldots,V_p\}$ of the set of
  $k$-ary input variables thus has $p=O(k^d)$. Claim: for each $i$,
  the best attainable lower bound on the number of states
  querying variables from $V_i$ is $O(k^{d-1})$.

Consider such a set $V_i$, $|V_i|=v\geq 1$.
Here $|R|=k^d$, so the number $\dkway(s,v)$ of distinct
deterministic BPs having $s$
non-final states querying variables from $V_i$ satisfies
\begin{displaymath}
\dkway(s,v) \geq 1^s \cdot (s+|R|)^{sk}
\geq (1+k^d)^{sk}
\geq k^{dsk}.
\end{displaymath}
Hence the estimate used in the Ne\u{c}iporuk method to upper bound
$\dkway(s,v)$ will be at least $k^{dsk}$.
On the other hand,
the number of functions $f_{V_i}:[k]^v \rightarrow R$ obtained
by fixing variables outside of $V_i$ cannot exceed
$k^{O(k^d)}$ since the number of variables outside $V_i$ is
$\Theta(k^d)$.
Hence the best lower bound on the number of states querying variables
from $V_i$ obtained by applying the method will be no
larger than the smallest $s$ verifying $k^{ck^d}\leq k^{dsk}$ for some $c$
depending on $d$ and $k$. This
proves our claim since then this number is at most $s = O(k^{d-1})$.
\end{proof}

Let \funcmod\ have the same input as \ft\ with the exception that the
root function is preset to the sum modulo $k$.
In other words the output is $v_2+ v_3+ \cdots + v_{d+1}\mod k$.

\begin{theorem} \label{t:lasttheorem}
The best $k$-way deterministic BP
size lower bound attainable for \funcmodtwothree\ by applying the
Ne\u{c}iporuk method is $\Omega(k^2)$.
\end{theorem}

\begin{proof}
  The function $\funcmodtwothree: [k]^m \rightarrow R$ has
  $m=\Theta(k^2)$. Consider a set $V_i$ in any partition
  $\{V_1,\ldots,V_p\}$ of the set of $k$-ary input variables, $|V_i|=v$.
Here $|R|=k$, so the number $\dkway(s,v)$ of distinct
deterministic BPs having $s$
non-sink states querying variables from $V_i$ satisfies
\begin{displaymath}
\dkway(s,v) \geq 1^s \cdot (s+|R|)^{sk}
\geq (1+k)^{sk}
\geq k^{sk}.
\end{displaymath}
If $V_i$ contains a leaf variable, then perhaps the number of
functions induced by setting variables complementary to $V_i$ can
reach the maximum $k^{k^2}$. Ne\u{c}iporuk would conclude that $k$
states querying the variables from such a $V_i$ are necessary.
Note that there are at most $4$ sets $V_i$ containing a leaf
variable (hence a total of $4k$ states required to account for the
variables in these $4$ sets).
Now suppose that $V_i$ does not contain a leaf variable.
Then setting the variables complementary to $V_i$ can either induce
a constant function (there are $k$ of those), or the sum of a constant
plus a variable (there are at most $k\cdot |V_i|$ of those)
or the sum of two of the variables (there are at most $|V_i|^2$ of
those). So the maximum number of induced functions is
$|V_i|^2=O(k^4)$. The number of states querying variables from $V_i$ is
found by Ne\u{c}iporuk to be $s\geq 4/k$. In other words $s=1$.
So for any of the at least $p-4$ sets in the partition not
containing a leaf variable, the method gets one state. Since
$p-4=O(k^2)$, the total number of states accounting for all the $V_i$
is $O(k^2)$.
\end{proof}

\subsection{The state sequence method}\label{s:beating}

Here we give alternative proofs for some of the lower bounds
given in Section \ref{s:NecLB}.  These proofs
are more intricate than the Ne\u{c}iporuk proofs but they do not
suffer a priori from a quadratic limitation.  The method also yields
stronger lower bounds for \functwofour\ and \funcmodtwothree\
(Theorems \ref{t:childLB} and \ref{t:beatittwice}) than
those obtained by applying Ne\u{c}iporuk's method (Theorems
\ref{t:rootfunction} and \ref{t:lasttheorem}).

\begin{theorem}\label{t:newLB}
$\ndtwothreeBstate \ge k^{2.5}$ for sufficiently large $k$.
\end{theorem}

\begin{proof}
Consider an input $I$ to \bttwothree.  We number the nodes
in $T^3_2$ as in Figure \ref{sample}, and let $v^I_j$ denote
the value of node $j$ under input $I$.  We say
that a state in a computation on input $I$ {\em learns} $v^I_j$ if that
state queries $f_j^I(v^I_{2j},v^I_{2j+1})$ (recall $2j,2j+1$ are the
children of node $j$).

\noindent
{\bf Definition [Learning Interval]}
\emph{Let $B$ be a $k$-way
nondeterministic BP that solves \bttwothree.
Let
${\cal C} = \gamma_0, \gamma_1, \cdots, \gamma_T$
be a computation of $B$ on input $I$.  We say that a state $\gamma_i$
in the computation is {\em critical} if one or more of the following holds:\begin{enumerate} \item $i = 0$ or $i=T$ \item $\gamma_i$ learns $v^I_2$
and there is an earlier state which learns $v^I_3$ with no intervening
state that learns $v^I_2$.  \item $\gamma_i$ learns $v^I_3$ and no
earlier state learns $v^I_3$ unless an intervening state learns $v^I_2$.
\end{enumerate}
We say that a subsequence $\gamma_i,\gamma_{i+1},\cdots
\gamma_j$ is a {\em learning interval} if $\gamma_i$ and $\gamma_j$ are
consecutive critical states.  The interval is {\em type 3} if $\gamma_i$
learns $v^I_3$, and otherwise the interval is {\em type 2}.
}

Thus type 2 learning intervals begin with $\gamma_0$ or a state
which learns $v^I_2$, and never learn $v^I_3$
until the last state, and type 3 learning intervals begin with a
state which learns $v^I_3$ and never learn $v^I_2$
until the last state.

Now let $B$ be as above, and for $j\in\{2,3\}$ let $\Gamma_j$
be the set of all states of $B$ which query the input function $f_j$.
We will prove the theorem by showing that for large $k$
\begin{equation}\label{e:proof}
  |\Gamma_2| + |\Gamma_3| > k^2\sqrt{k}.
\end{equation}

For $r,s\in[k]$ let $F^{r,s}_{yes}$ be the set of inputs $I$ to $B$
whose four leaves are labelled $r,s,r,s$ respectively,
whose middle node functions $f_2^I$ and $f_3^I$
are identically 1 except
$f^I_2(r,s)=v^I_2$ and $f^I_3(r,s)=v^I_3$, and $f^I_1(v^I_2,v^I_3)=1$
(so $v^I_1 = 1$).  Thus each such $I$ is a `YES input', and
should be accepted by $B$.

Note that each member $I$ of $F_{yes}^{r,s}$ is uniquely
specified by a triple
\begin{equation}\label{e:triple} (v^I_2,v^I_3,f^I_1)
\mbox{ where $f^I_1(v^I_2,v^I_3)=1$}
\end{equation}
and hence
$F_{yes}^{r,s}$ has exactly $k^2(2^{k^2-1})$ members.

For $j\in\{2,3\}$ and $r,s\in [k]$ let $\Gamma_j^{r,s}$ be the subset
of $\Gamma_j$ consisting of those states which query $f_j(r,s)$.
Then $\Gamma_j$ is the disjoint union of $\Gamma_j^{r,s}$ over all pairs
$(r,s)$ in $[k]\times [k]$.  Hence to prove (\ref{e:proof}) it suffices
to show
\begin{equation}\label{e:gammas}
  |\Gamma_2^{r,s}| + |\Gamma_3^{r,s}| > \sqrt{k}
\end{equation}
for large $k$ and all $r,s$ in $[k]$.  We will
show this by showing
\begin{equation}\label{e:product}
   (|\Gamma_2^{r,s}|+1) (|\Gamma_3^{r,s}|+1) \ge k/2
\end{equation}
for all $k\ge 2$.  (Note that given the product, the sum
is minimized when the summands are equal.)

For each input $I$ in $F_{yes}^{r,s}$ we associate a fixed
accepting computation ${\cal C}(I)$ of $B$ on input $I$.

Now fix $r,s\in [k]$.
For $a,b\in [k]$ and $f:[k]\times [k] \ra \{0,1\}$ with $f(a,b)=1$ we
use $(a,b,f)$ to denote the input $I$ in $F_{yes}^{r,s}$
it represents as in (\ref{e:triple}).

To prove (\ref{e:product}), the idea is that if it is false, then as
$I$ varies through all inputs $(a,b,f)$ in $F_{yes}^{r,s}$
there are too few states learning $v^I_2 = a$ and $v^I_3 =b$ to verify
that $f(a,b)=1$.  Specifically, we can find $a,b,f,g$ such that
$f(a,b)=1$ and $g(a,b)=0$,
and by cutting and pasting the accepting computation ${\cal C}(a,b,f)$
with accepting computations of the form ${\cal C}(a,b',g)$ and
${\cal C}(a',b,g)$ we can construct an accepting computation of
the `NO input' $(a,b,g)$.

We may assume that the branching program
$B$ has a unique initial state $\gamma_0$ and a unique accepting state
$\delta_{ACC}$.

For $j\in \{2,3\}$, $a,b\in [k]$ and $f:[k]\times [k]\ra \{0,1\}$ with
$f(a,b)=1$ define $\varphi_j(a,b,f)$ to be the set of all state pairs
$(\gamma,\delta)$ such that there is a type $j$ learning interval in
${\cal C}(a,b,f)$ which begins with $\gamma$ and ends with $\delta$.
Note that if $j=2$ then $\gamma\in(\Gamma_2^{r,s} \cup \{\gamma_0\})$
and $\delta \in (\Gamma_3^{r,s} \cup \{\delta_{ACC}\})$, and if $j=3$
then $\gamma\in\Gamma_3^{r,s}$ and $\delta \in (\Gamma_2^{r,s}
\cup \{\delta_{ACC}\})$.

To complete the definition, define $\varphi_j(a,b,f)=\varnothing$
if $f(a,b)=0$.

For $j\in\{2,3\}$ and $f:[k]\times [k] \ra \{0,1\}$ we define a
function $\varphi_j[f]$ from $[k]$ to sets of state pairs as follows:
\begin{eqnarray*}
  \varphi_2[f](a) & =  &\bigcup_{b\in[k]} \varphi_2(a,b,f) \
           \subseteq S_2 \\
  \varphi_3[f](b) & =  &\bigcup_{a\in[k]} \varphi_3(a,b,f) \
        \subseteq S_3
\end{eqnarray*}
where $S_2 = (\Gamma_2^{r,s}\cup \{\gamma_0\}) \times
             (\Gamma_3^{r,s} \cup \{\delta_{ACC}\})$ and
$S_3 = \Gamma_3^{r,s} \times
             (\Gamma_2^{r,s} \cup \{\delta_{ACC}\})$.

For each $f$ the function $\varphi_j[f]$ can be specified by listing a
$k$-tuple of subsets of $S_j$, and hence there are at most $2^{k|S_j|}$
distinct such functions as $f$ ranges over the $2^{k^2}$ Boolean functions
on $[k]\times[k]$, and hence there are at most $2^{k(|S_2|+|S_3|)}$
pairs of functions $(\varphi_2[f],\varphi_3[f])$.  If we assume that
(\ref{e:product}) is false, we have $|S_2|+|S_3| < k$.  Hence by the
pigeonhole principle there must exist distinct Boolean functions $f,g$
such that $\varphi_2[f] = \varphi_2[g]$ and $\varphi_3[f] = \varphi_3[g]$.

Since $f$ and $g$ are distinct we may assume that there exist $a,b$
such that $f(a,b)=1$ and $g(a,b)=0$.  Since $\varphi_2[f](a) =
\varphi_2[g](a)$, if $(\gamma,\delta)$ are the endpoints of a type 2
learning interval in ${\cal C}(a,b,f)$ there exists $b'$ such that
$(\gamma,\delta)$ are the endpoints of a type 2 learning interval
in ${\cal C}(a,b',g)$ (and hence $g(a,b')=1$).  Similarly, if
$(\gamma,\delta)$ are endpoints of a type 3 learning interval
in ${\cal C}(a,b,f)$ there exists $a'$ such that $(\gamma,\delta)$
are the endpoints of a type 3 learning interval in ${\cal C}(a',b,f)$.

Now we can construct an accepting computation for the `NO input'
$(a,b,g)$ from ${\cal C}(a,b,f)$ by replacing each learning
interval beginning with some $\gamma$ and ending with some $\delta$
by the corresponding learning interval in ${\cal C}(a,b',g)$
or ${\cal C}(a',b,g)$.  (The new accepting computation has the same
sequence of critical states as ${\cal C}(a,b,f)$.)
This works because a type 2 learning interval
never queries $v_3$ and a type 3 learning interval never queries $v_2$.

This completes the proof of (\ref{e:product}) and the theorem.
\end{proof}

\begin{theorem}\label{t:detThree} Every deterministic branching program
that solves \bttwothree\ has at least $k^3/\log k$ states for
sufficiently large $k$.
\end{theorem}

\begin{proof} We modify the proof of Theorem \ref{t:newLB}.  Let $B$
be a deterministic BP which solves \bttwothree, and for $j\in\{2,3\}$ let
$\Gamma_j$ be the set of states in $B$ which query $f_j$ (as before).
It suffices to show that for sufficiently large $k$
\begin{equation}\label{e:gammaSum}
  |\Gamma_2|+|\Gamma_3|\ge k^3/\log k.
\end{equation}

For $r,s \in [k]$ we define the set $F^{r,s}$ to be the same as
$F_{yes}^{r,s}$ except that we remove the restriction on $f^I_1$.
Hence there are exactly $k^2 2^{k^2}$ inputs in $F^{r,s}$.

As before, for $j\in\{2,3\}$, $\Gamma_j$ is the disjoint union
of $\Gamma^{r,s}$ for $r,s\in [k]$.
Thus to prove (\ref{e:gammaSum}) it suffices to show that
for sufficiently large $k$ and all $r,s$ in $[k]$
\begin{equation}\label{e:newGsum}
   |\Gamma_2^{r,s}| + |\Gamma_3^{r,s}| \ge k/\log k.
\end{equation}
We may assume there are unique start, accepting, and rejecting states
$\gamma_0$, $\delta_{ACC}$, $\delta_{REJ}$.
Fix $r,s\in[k]$.

For each root function $f:[k]\times[k]\ra \{0,1\}$ we define the functions
\begin{eqnarray*}
   \psi_2[f] : [k]\times (\Gamma_2^{r,s} \cup \{\gamma_0\}) & \ra &
           (\Gamma_3^{r,s} \cup \{\delta_{ACC},\delta_{REJ}\})\\
   \psi_3[f] : [k]\times \Gamma_3^{r,s}
      & \ra & (\Gamma_2^{r,s} \cup \{\delta_{ACC},\delta_{REJ}\})
\end{eqnarray*}
by $\psi_2[f](a,\gamma) = \delta$ if $\delta$ is the
next critical state after $\gamma$ in a computation with input $(a,b,f)$
(this is independent of $b$), or $\delta=\delta_{REJ}$ if there is no
such critical state.  Similarly $\psi_3[f](b,\delta)=\gamma$ if $\gamma$
is the next critical state after $\delta$ in a computation with input
$(a,b,f)$ (this is independent of $a$), or $\delta=\delta_{REJ}$ if
there is no such critical state.

\bigskip

\noindent CLAIM:  The pair of functions $(\psi_2[f],\psi_3[f])$ is
distinct for distinct $f$.

\bigskip

For suppose otherwise.  Then there are $f,g$ such that
$\psi_2[f]=\psi_2[g]$ and $\psi_3[f]=\psi_3[g]$ but $f(a,b)\ne g(a,b)$
for some $a,b$.  But then the sequences of critical states in the two
computations $C(a,b,f)$ and $C(a,b,g)$ must be the same, and hence the
computations either accept both $(a,b,f)$ and $(a,b,g)$ or reject both.
So the computations cannot both be correct.

Finally we prove (\ref{e:newGsum}) from the CLAIM.  Let $s_2 =
|\Gamma_2^{r,s}|$ and let $s_3=|\Gamma_3^{r,s}|$, and let $s=s_2+s_3$.
Then the number of distinct pairs $(\psi_2,\psi_3)$ is at most $$
   (s_3+2)^{k(s_2+1)}(s_2+2)^{ ks_3} \le (s+2)^{k(s+1)}
$$ and since there are $2^{k^2}$ functions $f$ we have $$
   2^{k^2} \le (s+2)^{k(s+1)}
$$ so taking logs, $k^2 \le k(s+1)\log (s+2)$ so $k/\log(s+2) \le
s+1$, and (\ref{e:newGsum}) follows.
\end{proof}

Recall from Theorem~\ref{t:rootfunction} that applying the
Ne\u{c}iporuk method to \functwofour\ yields an $\Omega(k^3)$ size
lower bound and from Theorem~\ref{t:lasttheorem} that applying it to
\funcmodtwothree\ yields $\Omega(k^2)$.  The next two results
improve on these bounds using the state sequence method.
The new lower bounds match the upper bounds given by the pebbling
method used to prove (\ref{e:dFUpper}) in Theorem \ref{t:BPUpper}.  

\begin{theorem}\label{t:childLB}
Any deterministic $k$-way BP for \functwofour\
has at least $k^4/2$ states.
\end{theorem}

\begin{proof}
Let $E_{4true}$ be the set of all inputs $I$ to \functwofour\
such that $f_2^I=f^I_3 =
+_k$ (addition mod $k$), and for $i\in\{4,5,6,7\}$ $f^I_i$ is identically
0 except for $f^I_i(v^I_{2i},v^I_{2i+1})$.

Let $B$ be a branching program as in the theorem.  For each $I\in
E_{4true}$ let ${\cal C}(I)$ be the computation of $B$
on input $I$.

For $r,s\in[k]$ let $E^{r,s}_{4true}$ be
the set of inputs $I$ in $E_{4true}$
such that for $i\in\{4,5,6,7\}$, $v^I_{2i}=r$ and $v^I_{2i+1}=s$.
Then for each pair $r,s$ each input $I$ in $E^{r,s}_{4true}$ is
completely specified by the quadruple $v^I_4,v^I_5,v^I_6,v^I_7$, so
$|E^{r,s}_{4true}| = k^4$.

For $r,s\in[k]$ and $i\in\{4,5,6,7\}$ let $\Gamma^{r,s}_i$
be the set of states of $B$ that query $f_i(r,s)$, and let
\begin{equation}\label{e:Gammars}
     \Gamma^{r,s} = \Gamma^{r,s}_4\cup\Gamma^{r,s}_5\cup\Gamma^{r,s}_6
     \cup
          \Gamma^{r,s}_7
\end{equation}

The theorem follows from the following Claim.

\medskip
\noindent
CLAIM 1:  $|\Gamma^{r,s}|  \ge k^2/2$
for all $r,s \in [k]$.

\medskip

To prove CLAIM 1, suppose to the contrary for some $r,s$
\begin{equation}\label{e:gammak}
    |\Gamma^{r,s}| < k^2/2
\end{equation}

We associate a pair
$$
  T(I) = (\gamma^I,v^I_i)
$$
with $I$ as follows: $\gamma^I$ is the last state in the
computation ${\cal C}(I)$ that is in $\Gamma^{r,s}$ (such a state
clearly exists), and
$i\in\{4,5,6,7\}$ is the node queried by $\gamma^I$.
(Here $v^I_i$ is the value of node $i$).

We also associate a second triple $U(I)$ with each input $I$ in
$E^{r,s}_{4true}$ as follows: $$
  U(I) = \left\{ \begin{array}{ll}
               (v^I_4,v^I_5,v^I_3) &
      \mbox{if $\gamma^I$ queries node 4 or 5 } \\
               (v^I_6,v^I_7,v^I_2) & \mbox{otherwise.}
              \end{array}
\right.  $$
 \medskip
\noindent
CLAIM 2: As $I$ ranges over $E^{r,s}_{4true}$, $U(I)$
ranges over at least $k^3/2$ triples in $[k]^3$.

\medskip

To prove CLAIM 2, consider the the subset $E'$ of inputs in
$E^{r,s}_{4true}$ whose values for nodes 4,5,6,7 have the form $a,b,a,c$
for arbitrary $a,b,c \in [k]$.  For each such $I$ in $E'$ an adversary
trying to minimize the number of triples $U(I)$ must choose one of the
two triples $(a,b,a+_k c)$ or $(a,c,a+_k b)$.  There are a total of $k^3$
distinct triples of each of the two forms, and the adversary must choose
at least half the triples from one of the two forms, so there must be at
least $k^3/2$ distinct triples of the form $U(I)$.   This proves CLAIM 2.

On the other hand by (\ref{e:gammak}) there are fewer than $k^3/2$
possible values for $T(I)$.  Hence there exist inputs $I,J \in
E^{r,s}_{4true}$ such that $U(I) \ne U(J)$ but $T(I)=T(J)$.  Since
$U(I) \ne U(J)$ but $v^I_i = v^J_i$ (where $i$ is the node queried by
$\gamma^I=\gamma^J$) it follows that either $v^I_2 \ne v^J_2$ or $v^I_3
\ne v^J_3$, so $I$ and $J$ give different values to the function
\functwofour.  But since $T(I)=T(J)$ if follows that the two
computations ${\cal C}(I)$ and ${\cal C}(J)$ are in the same
state $\gamma^I=\gamma^J$ the last time any of the nodes
$\{4,5,6,7\}$ is queried, and the answers $v^I_i=v^J_i$
to the queries are the same, so both computations give identical
outputs.  Hence one of them is wrong.
\end{proof}

\begin{theorem}\label{t:beatittwice}
Any deterministic $k$-way BP for \funcmodtwothree\ requires
at least $k^3$ states.
\end{theorem}

\begin{proof}
We adapt the previous proof.  Now $E^{r,s}$ is the set of
inputs $I$ to \funcmodtwothree\ such that for $i\in\{2,3\}$,
$f^I_i$ is identically one except possibly for $f^I_i(r,s)$,
and $v^I_4 = v^I_6 = r$ and $v^I_5 = v^I_7 = s$.
Note that an input to $E^{r,s}$ can be specified by the
pair $(v^I_2,v^I_3)$, so $E^{r,s}$ has exactly $k^2$ elements.
Define
$$
   \Gamma^{r,s} = \Gamma^{r,s}_2 \cup \Gamma^{r,s}_3
$$
Now we claim that an input $I$ in $E^{r,s}$ can be specified
by the pair $(\gamma^I,v^I_i)$, where $\gamma^I$ is the last
state in the computation ${\cal C}(I)$ that is in $\Gamma^{r,s}$,
and $i\in\{2,3\}$ is the node queried by $\gamma^I$.

The Claim holds because $(\gamma^I,v^I_i)$ determines the output
of the computation, which in turn (together with $v^I_i$)
determines $v^I_j$, where $j$ is the sibling of $i$.

From the Claim it follows that $|\Gamma^{r,s}| \ge k$ for
all $r,s\in [k]$, and hence there must be at least $k^3$ states
in total.
\end{proof}

\subsection{Thrifty lower bounds}
\label{s:thriftyLB}

See Definition \ref{d:thrifty} for thrifty programs.

Theorem \ref{t:detThriftLB} below shows that the upper bound
given in Theorem \ref{t:BPUpper} (\ref{e:dFUpper}) is optimal
for deterministic thrifty programs solving the function
problem \ft\ for $d=2$ and all $h\ge 2$.
Theorem \ref{t:thrifFourtwo} shows that the upper bound given
in Theorem \ref{t:BPUpper} (\ref{e:nBUpper}) is optimal for
nondeterministic thrifty programs solving the Boolean problem
\bt\ for $d=2$ and $h=4$ (it is optimal for $h \le 3$ by 
Theorem \ref{c:HtThree}).

\begin{theorem}\label{t:detThriftLB}
For any $h,k$, every deterministic thrifty branching program solving
$BT_2^{h}(k)$ has at least $k^h$ states.
\end{theorem}

Fix a deterministic thrifty BP $B$ that solves $BT_2^{h}(k)$. Let $E$ be the inputs to $B$. Let $\Vars$ be the set of $k$-valued input variables (so $|E| = k^{|\Vars|}$). Let $Q$ be the states of $B$. If $i$ is an internal node then the $i$ variables are $f_i(a,b)$ for $a,b \in [k]$, and if $i$ is a leaf node then there is just one $i$ variable $l_i$. We sometimes say ``$f_i$ variable'' just as an in-line reminder that $i$ is an internal node. Let $\var(q)$ be the input variable that $q$ queries. Let $\node$ be the function that maps each variable $X$ to the node $i$ such that $X$ is an $i$ variable, and each state $q$ to $\node(\var(q))$. When it is clear from the context that $q$ is on the computation path of $I$, we just say ``$q$ queries $i$'' instead of ``$q$ queries the thrifty $i$ variable of $I$''.

Fix an input $I$, and let $P$ be its computation path. We will choose $n$ states on $P$ as {\bf critical states} for $I$, one for each node. 
Note that $I$ must visit a state that queries the root (i.e. queries the thrifty root variable of $I$), since otherwise the branching program would make a mistake on an input $J$ that is identical to $I$ except 
$f_1^J(v_2^I,v_3^I) := k - f_1^I(v_2^I,v_3^I)$;
 hence $J \in BT^h_2(k)$ iff $I \not \in BT^h_2(k)$. So, we can choose the root critical state for $I$ to be the last state on $P$ that queries the root. The remainder of the definition relies on the following small lemma:
\begin{lemma}\label{l:basic_thrifty}
 For any $J$ and internal node $i$, if $J$ visits a state $q$ that queries $i$, then for each child $j$ of $i$, there is an earlier state on the computation path of $J$ that queries $j$.
\end{lemma}
\begin{proof}
 Suppose otherwise, and wlog assume the previous statement is false for $j=2i$. For every $a \not = v_{2i}^J$ there is an input $J_a$ that is identical to $J$ except $v_{2i}^{J_a} = a$.
But the computation paths of $J_a$ and $J$ are identical up to $q$, so $J_a$ queries a variable $f_i(a,b)$ such that $b = v_{2i+1}^{J_a}$ and $a \not = v_{2i}^{J_a}$, which contradicts the thrifty assumption.
\end{proof} 

Now we can complete the definition of the critical states of $I$. For $i$ an internal node, if $q$ is the node $i$ critical state for $I$ then the node $2i$ (resp. $2i+1$) critical state for $I$ is the last state on $P$ before $q$ that queries $2i$ (resp. $2i+1$).

We say that a collection of nodes is a {\em minimal cut} of the tree if 
every path from root to leaf contains exactly one of the nodes.
Now we assign a pebbling sequence to each state on $P$, such that the set of pebbled nodes in each configuration is a minimal cut of the tree or a subset of some minimal cut (and once it becomes a minimal cut, it remains so), and any two adjacent configurations are either identical, or else the later one follows from the earlier one by a valid pebbling move.  (Here we
allow the removal of the pebbles on the children of a node $i$ as part
of the move that places a pebble on $i$.)  This assignment can be described inductively by starting with the last state on $P$ and working backwards. Note that implicitly we will be using the following fact:

\begin{fact}\label{f:basic_crit_state}
For any input $I$, if $j$ is a descendant of $i$ then the node $j$ critical state for $I$ occurs earlier on the computation path of $I$ than the node $i$ critical state for $I$.  
\end{fact}

The pebbling configuration for the output state has just a black pebble on the root. Assume we have defined the pebbling configurations for $q$ and every state following $q$ on $P$, and let $q'$ be the state before $q$ on $P$. If $q'$ is not critical, then we make its pebbling configuration be the same as that of $q$. If $q'$ is critical then it must query a node $i$ that is pebbled in $q$. The pebbling configuration for $q'$ is obtained from the configuration for $q$ by removing the pebble from $i$ and adding pebbles to $2i$ and $2i+1$ (if $i$ is an internal node - otherwise you only remove the pebble from $i$). 

Now consider the last critical state in the computation path $P^I$ that
queries a height 2 node (i.e. a parent of leaves).  
We use $r^I$ to denote this state and call it the
{\bf supercritical state} of $I$.  The pebbling configuration 
associated with $r^I$ is called the bottleneck configuration, and its
pebbled nodes are called {\bf bottleneck nodes}.  The two children of
$\node(r^I)$ must be bottleneck nodes, and the bottleneck nodes
form a minimal cut of the tree.   The path from the root to
$\node(r)$ is the {\bf bottleneck path}, and by Fact
\ref{f:basic_crit_state} it cannot contain any bottleneck nodes.
From all this it is easy to see that there must be at least $h$
bottleneck nodes.

Here is the main property of the pebbling sequences that we need: 

\begin{fact}\label{f:basic_peb_seq}
For any input $I$, if non-root node $i$ with parent $j$ is pebbled at a state $q$ on $P^I$, then the node $j$ critical state $q'$ of $I$ occurs later on $P^I$, and there is no state (critical or otherwise) between $q$ and $q'$ on $P^I$ that queries $i$. 
\end{fact}
Let $R$ be the states that are supercritical for at least one input. Let $E_r$ be the inputs with supercritical state $r$. Now we can state the main lemma. 
\begin{lemma}\label{l:thrifty_advice_main_lemma}
 For every $r \in R$, there is an surjective function 
from $[k]^{|\Vars|-h}$ to $E_r$.
\end{lemma}
The lemma gives us that $|E_r| \le k^{|\Vars|-h}$ for every $r \in R$. Since $\{E_r\}_{r\in R}$ is a partition of $E$, there must be at least $|E| / k^{|\Vars|-h} = k^h$ sets in the partition, i.e. there must be at least $k^h$ supercritical states. So the theorem follows from the lemma. 

\begin{proof}
Fix $r \in R$ and let $D := E_r$. Let $\isc := \node(r)$. Since $r$ is
thrifty for every $I$ in $D$, there are values $v_{2\isc}^D$ and
$v_{2\isc+1}^D$ such that $v_{2\isc}^I = v_{2\isc}^D$ and
$v_{2\isc+1}^I = v_{2\isc+1}^D$ for every $I$ in $D$.   The surjective
function of the lemma is computed by a procedure $\IntAdv$ that takes
as input a $[k]$-string (the advice), tries to interpret it as the code
of an input in $D$, and when successful outputs that input. We want to
show that for every $I \in D$ we can choose $\adv^I \in [k]^{|\Vars|-h}$
such that $\IntAdv(\adv^I)\DEF = I$.

The idea is that
the procedure $\IntAdv$ traces the computation path $P$ starting from
state $r$, using the advice string $\adv^I$ when necessary to answer
queries made by each state $q$ along the path.  By the thrifty
property, the procedure can `learn' the values $a,b$ of the children of
$i=\node(q)$ (if $i$ is an internal node) from the query $f_i(a,b)$ of $q$.
Each such child that has not been queried earlier in the trace saves
one advice value for the future.  By Fact \ref{f:basic_peb_seq} the
parent of each of the $h$ bottleneck nodes will be queried
before the node itself, making a total savings of at least $h$
values in the advice string.  After the trace is completed, the
remaining advice values complete the specification of the input
$I\in E_r$.

In more detail, during the execution of the procedure we maintain a current state $q$, a partial function $v^*$ from nodes to $[k]$, and a set of nodes $\UL$. Once we have added a node to $\UL$, we never remove it, and once we have added $v^*(i) := a$ to the definition of $v^*$, we never change $v^*(i)$. We have reached $q$ by following a \emph{consistent partial computation path} starting from $r$, meaning there is at least one input in $D$ that visits exactly the states and edges that we visited between $r$ and $q$. So initially $q = r$. Intuitively, $v^*(i)\DEF = a$ for some $a$ when we have ``committed'' to interpreting the advice we have read so-far as being the initial segment of \emph{some} complete advice string $\adv^I$ for an input $I$ with $v_i^I = a$. Initially $v^*$ is undefined everywhere. As the procedure goes on, we may often have to use an element of the advice in order to set a value of $v^*$; however, by exploiting the properties of the critical state sequences, for each $I \in D$, when given the complete advice $\adv^I$ for $I$ there will be at least $h$ nodes $\UL^I$ that we ``learn'' without directly using the advice. Such an opportunity arises when we visit a state that queries some variable $f_i(b_1,b_2)$ and we have not yet committed to a value for at least one of $v^*(2i)$ or $v^*(2i+1)$ (if both then, we learn two nodes). When this happens, we add that child or children of $i$ to $\UL$ (the {\sf L} stands for ``learned''). So initially $\UL$ is empty. There is a loop in the procedure $\IntAdv$ that iterates until $|\UL| = h$. Note that the children of $\isc$ will be learned immediately. Let $v^*(D)$ be the inputs in $D$ consistent with $v^*$, i.e. $I \in v^*(D)$ iff $I \in D$ and $v_i^I = v^*(i)$ for every $i \in \Dom(v^*)$. 

Following is the complete pseudocode for $\IntAdv$. We also state the most-important of the invariants that are maintained. \\

\noindent
{\bf Procedure} $\IntAdv(\vec{a} \in [k]^*)$:
\begin{algorithmic}[1]
  \STATE $q := r$,\ $\UL := \emptyset$,\ $v^* := \text{undefined everywhere}$.
  \STATE {\bf Loop Invariant:} If $N$ elements of $\vec{a}$ have been used, then $|\Dom(v^*)| = N + |\UL|$.\\
  \WHILE{$|\UL| < h$}
  	  \STATE $i := \node(q)$
        \IF{$i$ is an internal node and $2i \not \in \Dom(v^*)$ or $2i+1 \not \in \Dom(v^*)$}
	        \STATE let $b_1,b_2$ be such that $\var(q) = f_i(b_1,b_2)$.
	        \IF{$2i \not \in \Dom(v^*)$}
	        	\STATE $v^*(2i) := b_1$ and $\UL := \UL + 2i$.
		  \ENDIF
		  \IF{$2i+1 \not \in \Dom(v^*)$ and $|\UL| < h$}
		  	\STATE $v^*(2i+1) := b_2$ and $\UL := \UL + (2i+1)$.
	        \ENDIF
        \ENDIF
        \IF{$i \not \in \Dom(v^*)$}
        	\STATE let $a$ be the next unused element of $\vec{a}$.
	       \STATE $v^*(i) := a$.
        \ENDIF
        \STATE $q := $ the state reached by taking the edge out of $q$ labeled $v^*(i)$.
  \ENDWHILE
  \STATE let $\vec{b}$ be the next $|\Vars| - |\Dom(v^*)|$ unused elements of $\vec{a}$. \label{line:lastadviceuse}
  \STATE let $I_1,\ldots,I_{|v^*(D)|}$ be the inputs in $v^*(D)$ sorted according to some globally fixed order on $E$. \label{line:secondtolastline}
  \STATE if $\vec{b}$ is the $t$-largest string in the lexicographical ordering of $[k]^{|\Vars| - |\Dom(v^*)|}$, and $t \le |v^*(D)|$, then return $I_t$.\footnote{See after this code for argument that $|v^*(D)| \le k^{|\Vars| - |\Dom(v^*)|}$.}  \label{line:lastline}
\end{algorithmic}

\vspace{10pt}

If the loop finishes, then there are at most $|E|/|\Dom(v^*)| = k^{|\Vars|-|\Dom(v^*)|}$ inputs in $v^*(D)$. 
So for each of the 
inputs $I$ enumerated on line \ref{line:secondtolastline}, there is a way of setting $\vec{a}$ so that $I$ will be chosen on line \ref{line:lastline}. 

Recall we are trying to show that for every $I$ in $D$ there is a string $\adv^I  \in [k]^{|\Vars|-h}$ such that $\IntAdv(\vec{a})\DEF = I$. This is easy to see under the assumption that there is such a string that makes the loop finish while maintaining the loop invariant; since the loop invariant ensures we have used $|\Dom(v^*)| - h$ elements of advice when we reach line \ref{line:lastadviceuse}, and since line \ref{line:lastadviceuse} is the last time when the advice is used, in all we use at most $|\Vars| - h$ elements of advice.
 To remove that assumption, first observe that for each $I$, we can set the advice to some $\adv^I$ so that $I \in g(D)$ is maintained when $\IntAdv$ is run on $\vec{a}^I$. Moreover, for that $\adv^I$, we will never use an element of advice to set the value of a bottleneck node of $I$, and $I$ has at least $h$ bottleneck nodes. Note, however, that this does not necessarily imply that $\UL^I$ (the $h$ nodes $\UL$ we obtain when running $\IntAdv$ on $\adv^I$) is a subset of the bottleneck nodes of $I$. Finally, note that we are of course implicitly using the fact that no advice elements are ``wasted''; each is used to set a different node value.
\end{proof}

\begin{cor}
 For any $h,k$, every deterministic thrifty branching program solving
$BT_2^{h}(k)$ has at least $\sum_{2 \le l \le h} k^l$ states.
\end{cor}
\begin{proof}
 The previous theorem only counts states that query height 2 nodes. The same proof is easily adapted to show there are at least $k^{h-l+2}$ states that query height $l$ nodes, for $l = 2,\ldots,h$.
\end{proof}

\begin{theorem} \label{t:thrifFourtwo}
Every nondeterministic thrifty branching program solving
\bttwofour\ has $\Omega(k^3)$ states.
\end{theorem}

\begin{proof}
As in the proof of the previous theorem  we restrict
attention to inputs $I$ in which the function $f_i$ associated with
each internal node $i$ (except $i=1$) satisfies $f_i(x,y)=0$
except possibly when $x,y$ are the values of its children.
For $r,s\in[k]$ let $E^{r,s}$ be the set of all such inputs $I$
such that for all $j\in\{4,5,6,7\}$, $v^I_{2j}=r$ and $v^I_{2j+1}=s$
(i.e. each pair of sibling leaves have values $r,s$), and
$f_1$ is identically 1 (so $I$ is a YES instance).
Thus $I$ is determined by the values of its 6 middle nodes
$\{2,3,4,5,6,7\}$, so
$$
          |E^{r,s}|= k^6
$$
Let $B$ be a nondeterministic thrifty branching program that solves
$T_2(4,k)$, and let $\Gamma$ be the set of states of $B$ which
query one of the nodes $4,5,6,7$.  We will show $|\Gamma| = \Omega(k^3)$.

For $r,s\in [k]$ let $\Gamma^{r,s}$ be the set of states of $\Gamma$
that query $f_j(r,s)$ for some $j\in\{4,5,6,7\}$.  We will show
\begin{equation}\label{e:GamLB}
  |\Gamma^{r,s}| +1 \ge k/\sqrt{3}
\end{equation}
Since $\Gamma$ is the disjoint union of $\Gamma^{r,s}$ for all
$r,s\in [k]$, it will follow that $|\Gamma|=\Omega(k^3)$ as required.

For each $I\in E^{r,s}$ let ${\cal C}(I)$ be an accepting
computation of $B$ on input $I$.  Let $t^I_1$ be the first time
during ${\cal C}(I)$ that the root $f_1$ is queried.  Let $\gamma^I$ be
be the last state in $\Gamma^{r,s}$ before $t^I_1$ in ${\cal C}(I)$
(or the initial state $\gamma_0$ if there is no such state) and
let $\delta^I$ be the first state in $\Gamma^{r,s}$ after $t^I_1$
(or the ACCEPT state $\delta_{acc}$ if there is no such state).

We associate with each $I\in E^{r,s}$ a tuple
$$
   U(I) = (u,\gamma^I,\delta^I,x_1,x_2,x_3,x_4)
$$
where $u\in\{1,2,3\}$ is a tag, and $x_1,x_2,x_3,x_4$ are in $[k]$
and are chosen so that
$U(I)$ uniquely determines $I$ (by determining the values of
all 6 middle nodes).
Specifically, $x_1 = v^I_i$,
where $i$ is the node queried by $\gamma^I$ (or $i=4$ if
$\gamma^I=\gamma_0$).

We partition $E^{r,s}$ into three sets $E^{r,s}_1,E^{r,s}_2,E^{r,s}_3$
according to which of the nodes $v_2,v_3$ the computation
${\cal C}(I)$ queries during the segment of the computation
between $\gamma^I$ and $\delta^I$.   (The tag $u$ tells us that $I$
lies in set $E^{r,s}_u$.)

Let node $j\in\{2,3\}$ be the parent of node $i$
(where $i$ is defined above) and let $j'\in\{2,3\}$ be the
sibling of $j$.

\begin{itemize}
\item
$E^{r,s}_1$ consists of those inputs $I$
for which ${\cal C}(I)$ queries neither $v_2$ nor $v_3$.

\item
$E^{r,s}_2$ consists of those inputs $I$ for which
${\cal C}(I)$ queries $v_{j'}$.

\item
$E^{r,s}_3$ consists of those inputs $I$ for which
${\cal C}(I)$ queries $v_j$ but not $v_{j'}$.
\end{itemize}

To complete the definition of $U(I)$ we need only specify
the meaning of $x_2,x_3,x_4$.

Let ${\cal S}(I)$ denote the segment of the computation
${\cal C}(I)$ between $\gamma^I$ and $\delta^I$ (not counting
the action of the last state $\delta^I$).   This
segment always queries the root $f_1(v_2,v_3)$, but
does not query any of the nodes
$4,5,6,7$ except $\gamma^I$ may query node $i$.

The idea is that the segment ${\cal S}(I)$ will determine
(using the definition of {\em thrifty}) the values of (at least) two of
the six middle nodes, and $x_1,x_2,x_3,x_4$ will specify the remaining
four values.  We require that
$x_1,x_2,x_3,x_4$ must specify the value of any node
(except the root) that is queried during the segment, but
the state that queries the node determines the values of its children.

In case the tag $u=1$, the computation queries $f_1(v_2,v_3)$,
and hence determines $v_2,v_3$, so $x_1,x_2,x_3,x_4$ specify
the four values $v_4,v_5,v_6,v_7$.

In case $u=2$, the computation queries $f_{j'}$ at the values
of its children, so $x_1,x_2,x_3,x_4$ do not specify the values
of these children, but instead specify $v_2,v_3$.

In case $u=3$, $x_1,x_2,x_3,x_4$ do not specify the value of
the sibling of node $i$ and do not specify $v_{j'}$, but
do specify $v_j$ and the values of the other level 2 nodes.

\medskip
\noindent
{\bf Claim:}  If $I,J\in E^{r,s}$ and $U(I) = U(J)$, then
$I=J$.

Inequality (\ref{e:GamLB}) (and hence the theorem) follows from
the Claim, because if $|\Gamma^{r,s}|+1< k/\sqrt{3}$ then
there would be fewer than $k^6$ choices for $U(I)$ as $I$ ranges
over the $k^6$ inputs in $E^{r,s}$.

To prove the Claim, suppose $U(I)=U(J)$ but $I\ne J$.  Then we can define
an accepting computation of input $I$ which violates the
definition of thrifty.  Namely follow the computation
${\cal C}(I)$ up to $\gamma^I$.  Now follow the segment of
${\cal C}(J)$ between $\gamma^I$ and $\delta^I$, and complete
the computation by following ${\cal C}(I)$.  Notice that
the segment of ${\cal C}(J)$ never queries any of the nodes
$4,5,6,7$ except for $v_i$, and $U(I) = U(J)$
(together with the definition of $E^{r,s}$) specifies the
values of the other nodes that it queries.  However,
since $I\ne J$, this segment of ${\cal C}(J)$ with input $I$
will violate the definition of {\em thrifty} while querying
at least one of the three nodes $v_1,v_2,v_3$.
\end{proof}

\section{Conclusion}\label{s:conclu}
The Thrifty Hypothesis (page \pageref{thriftyH}) states that
thrifty branching programs are optimal among $k$-way BPs solving
\ft.  For the deterministic case, this says that the black pebbling
method is optimal.  Proving this would separate \lspace\ from \p\
(Corollary \ref{c:thegoal}).
Even disproving this would be interesting, since it would show
that one can improve upon this obvious application of pebbling.

The next important step is to extend the tight branching program
bounds given in Corollary \ref{c:HtThree} for height 3 trees
to height 4 trees.  The upper bound given in Theorem \ref{t:BPUpper}
(\ref{e:dFUpper}) for the height 4 function problem $FT^4_d(k)$
for deterministic BPs is $O(k^{3d-2})$.  If we could match this
with a similar lower bound when $d=4$ (e.g. by using a variation
of the state sequence method in Section \ref{s:beating})
this  would yield $\Omega(k^{10})$
states for the function problem and hence (by Lemma \ref{l:FvsB})
$\Omega(k^9)$ states for the Boolean problem $BT^4_4(k)$.  This
would break the Ne\u{c}iporuk $\Omega(n^2)$ barrier for branching
programs (see Section \ref{s:NecLB}).

For nondeterministic BPs, the upper bound given by Theorem \ref{t:BPUpper}
for the Boolean problem for height 4 trees is $O(k^{2d-1})$.  This
comes from the upper bound on fractional pebbling given in
Theorem \ref{t:daryFract}, which we suspect is optimal for $h=4$
and degree $d=3$.  The corresponding lower bound for nondeterministic BPs
for $BT^4_3(k)$ would be $\Omega(k^5)$.  A proof would break
the Ne\u{c}iporuk $\Omega(n^{3/2})$ barrier for nondeterministic BPs.

Other (perhaps more accessible) open problems are to generalize
Theorem \ref{t:thrifFourtwo} to get general lower bounds for
nondeterministic thrifty BPs solving $BT^h_2(k)$, and to improve
Theorem \ref{t:daryFract} to get tight bounds on the number
of pebbles required to fractionally pebble $T^h_d$.

The proof of Theorem \ref{t:detThriftLB}, which states that
deterministic thrifty BPs require at least $k^h$ states to solve
$BT_2^{h}(k)$, is taken from \cite{wehr}.  That paper also proves
the same lower bound for the more general class of `less-thrifty' BPs,
which are allowed to query $f_i(a,b)$ provided that either $(a,b)$
correctly specify the values of both children of $i$, or neither $a$
nor $b$ is correct. 

\cite{wehr} also calculates $(k+1)^h$ as the exact
number of states required to solve $FT_2^h(k)$ using the
black pebbling method, and proves this is optimal when $h=2$.
So far we have not been able to beat this
BP upper bound by even one state, for any $h$ and any $k$
using any method.   That this bound might actually be unbeatable
(at least for all $h$ and all sufficiently large $k$) makes an
intriguing hypothesis.


\medskip
\noindent
{\bf Acknowledgment}  James Cook played a helpful role in the early
parts of this research.

\bibliographystyle{alpha}
\bibliography{paper}

\end{document}